\documentclass[american,aps,prx,reprint,superscriptaddress,longbibliography, floatfix,nofootinbib]{revtex4-2}
\usepackage[T1]{fontenc}
\usepackage{inputenc}
\usepackage{color}
\usepackage{babel}
\usepackage{physics}
\usepackage{mathtools}
\usepackage{bbding}
\usepackage{pifont}
\usepackage{textcomp}
\usepackage{wasysym}
\usepackage{amsthm}
\usepackage{nicefrac}
\usepackage{wasysym}
\usepackage{dsfont}
\usepackage{amsmath}
\usepackage{bbm}
\usepackage{amssymb}
\usepackage{bbold}
\usepackage{graphicx}
\usepackage{enumitem}
\definecolor{darkblue}{rgb}{0.0, 0.0, 0.55}

\usepackage[colorlinks=true,
            allcolors=blue]{hyperref}

\hypersetup{
     colorlinks   = true,
     linkcolor    = blue,
     citecolor    = blue,
     urlcolor     = blue,
     }

\usepackage{appendix}
\renewcommand{\selectlanguage}[1]{}
\makeatletter
\@ifundefined{textcolor}{}
{%
	\definecolor{BLACK}{gray}{0}
	\definecolor{WHITE}{gray}{1}
	\definecolor{RED}{rgb}{1,0,0}
	\definecolor{GREEN}{rgb}{0,1,0}
	\definecolor{BLUE}{rgb}{0,0,1}
	\definecolor{CYAN}{cmyk}{1,0,0,0}
	\definecolor{MAGENTA}{cmyk}{0,1,0,0}
	\definecolor{YELLOW}{cmyk}{0,0,1,0}
}

\usepackage{babel}
\usepackage{txfonts}
\usepackage{colortbl}\definecolor{myurlcolor}{rgb}{0,0,0.7}

\usepackage[normalem]{ulem}
\usepackage{tikz}
\usepackage{comment}
\usetikzlibrary{circuits.ee.IEC}
\usetikzlibrary{trees}

\newtheorem{lemma}{Lemma}

\newtheorem{result}{Result}

\newtheorem{proposition}{Proposition}

\def\tr{\operatorname{tr}}

\def\>{\rangle}
\def\<{\langle}

\makeatletter
\newsavebox\myboxA
\newsavebox\myboxB
\newlength\mylenA
\newcommand*\xoverline[2][0.75]{%
    \sbox{\myboxA}{$\m@th#2$}%
    \setbox\myboxB\null
    \ht\myboxB=\ht\myboxA%
    \dp\myboxB=\dp\myboxA%
    \wd\myboxB=#1\wd\myboxA
    \sbox\myboxB{$\m@th\overline{\copy\myboxB}$}
    \setlength\mylenA{\the\wd\myboxA}
    \addtolength\mylenA{-\the\wd\myboxB}%
    \ifdim\wd\myboxB<\wd\myboxA%
       \rlap{\hskip 0.5\mylenA\usebox\myboxB}{\usebox\myboxA}%
    \else
        \hskip -0.5\mylenA\rlap{\usebox\myboxA}{\hskip 0.5\mylenA\usebox\myboxB}%
    \fi}
\makeatother

\usepackage{hyperref}

\begin{document}
\title{Quantum waste management: Utilizing residual states in quantum information processing}

\author{Karol Horodecki}\email{karol.horodecki@ug.edu.pl} 
\affiliation{Institute of Informatics, National Quantum Information Centre, Faculty of Mathematics, Physics and Informatics, University of Gda\'nsk, Wita Stwosza 57, 80-308 Gda\'nsk, Poland}
\affiliation{International Centre for Theory of Quantum Technologies, University of Gda\'nsk, Wita Stwosza 63, 80-308 Gda\'nsk, Poland}
\affiliation{ School of Electrical and Computer Engineering, Cornell University, Ithaca, New York 14850, USA}

\author{Chirag Srivastava}\email{chirag.srivastava@ug.edu.pl}
\affiliation{Institute of Informatics, National Quantum Information Centre, Faculty of Mathematics, Physics and Informatics, University of Gda\'nsk, Wita Stwosza 57, 80-308 Gda\'nsk, Poland}

\author{Leonard Sikorski }\email{leonard.sikorski@phdstud.ug.edu.pl}
\affiliation{Institute of Informatics, National Quantum Information Centre, Faculty of Mathematics, Physics and Informatics, University of Gda\'nsk, Wita Stwosza 57, 80-308 Gda\'nsk, Poland}

\author{Siddhartha Das} \email{das.seed@iiit.ac.in} 
\affiliation{Center for Security, Theory and Algorithmic Research, Centre for Quantum Science and Technology, International Institute of Information Technology, Hyderabad, Gachibowli, Telangana 500032, India}

\begin{abstract}
Quantum resource theories use distillation protocols to convert less resourceful states into fully resourceful ones. However, these protocols often also generate an additional, unused output—referred to as a residual.
We propose a framework for the quantum residual management, in which states discarded after a resource distillation protocol are repurposed as inputs for subsequent quantum information tasks. This approach extends conventional quantum resource theories by incorporating secondary resource extraction from residual states, thereby enhancing overall resource utility. As a concrete example, we investigate the distillation of private randomness from the residual states remaining after quantum key distribution (QKD). More specifically, we quantitatively show that after performing a well-known coherent Devetak-Winter protocol, one can locally extract private randomness from its residual. We further consider the Gottesman-Lo QKD protocol and provide the achievable rate of private randomness from the discarded states that are left after its performance. We also provide a formal framework that highlights a general principle for improving quantum resource utilization across sequential information processing tasks. 
\end{abstract}

\maketitle

{\it Introduction}.--- In quantum information science, quantum states and operations can possess unique characteristics that make them resourceful for performing specific information-processing tasks~\cite{Got98,HHHH09,KDWW21,DBB+25,BGD25b}. Various well-studied resource properties include entanglement~\cite{HHHH09}, nonlocality~\cite{BGS05,Kar21,GCK+24}, steering~\cite{GA15,UCN+20}, coherence~\cite{SSD+15,WY16,CH16}, non-stabilizerness (commonly termed magic)~\cite{VHGE14,HC17}, purity~\cite{HO12,TKEP17}, non-Gaussianity~\cite{AGPF18,RLFT21}, athermality~\cite{BHO+13}, activity~\cite{SDC21,SSC23}, and non-Markovianity~\cite{Wak17,GPG+25}, among others. For example, entanglement underpins secure quantum key distribution in the presence of quantum adversaries~\cite{DBWH21}; nonlocality is essential for tasks such as device-independent quantum key distribution and certified randomness generation~\cite{BCP+19}; coherence and magic are central to the power of quantum computation~\cite{VHGE14,HC17}; and activity is instrumental in work extraction processes~\cite{ABN04}. The pursuit and manipulation of such resources are not only of foundational interest but are also driven by practical imperatives. Because of intrinsic constraints imposed by quantum theory and current technological limitations, the preparation, availability, and control of these resourceful states are both challenging and critically important.

 The framework of quantum resource theories offers a systematic approach to quantitatively assess the resourcefulness of quantum states and operations. It also provides protocols for distilling pure states that possess maximal resource content from an initial supply of noisy, less resourceful states, using operations that do not increase the resourcefulness of the quantum system~\cite{HO12,CG19,LBT19}. There are resource theories with multiple distinct resources whose interconversion is discussed~\cite{SdS+20}, as well as resource theories involving the simultaneous distillation of entanglement and coherence~\cite{CH16}. Both of which differ from the main focus of this work, as we will see.
 
 Within any resource theory, states that contain no resource are termed free states, and operations that do not increase the resource content are known as free operations. During the process of distillation, the transformation from noisy to highly resourceful states inevitably generates quantum states that are discarded as byproducts. We refer to these byproducts as residual states of the resource distillation process. Such states are typically ignored or treated as byproducts, akin to garbage, without further use.

 \textit{Idea}.--- Our everyday experience unequivocally demonstrates that value is often subjective, a principle succinctly captured by the adage `one person's trash is another person's treasure.' This idea can be extended to quantum information processing, where efficient resource management is increasingly important. We propose a novel framework for efficiently utilizing residual quantum states in sequential resource distillation. We investigate and illustrate, through various examples, methods for repurposing quantum states that remain after the initial distillation of a resource. These `leftover' states, if they retain sufficient utility, can then be used to distill a different quantum resource, thereby optimizing overall resource usage efficiency in quantum information processing.

 This perspective aligns with ongoing efforts in the Quantum Energy Initiative~\cite{ISM17,Auffves2022,JM23,MY25}, which seeks to understand and optimize energetic and informational costs in quantum technologies~\cite{HWS+25,BGD25b}. Related developments include  studies on recycling of nonlocality~\cite{Silva15,Curchod17,Brown20,Mal16,Steffinlongo2022,Sasmal2024} and entanglement~\cite{SHO13,Roy2021,PSS22,SPS25,MSSS25,Das25,kopszak25}, sample-efficient work extraction from unknown sources~\cite{CDG+25}. There also exists a work that enhances a state conversion process by reprocessing copies of states for which the conversion initially failed~\cite{Morelli22}. Note that all of these studies focus on sequential attempts to witness or improve the efficiency of a single resource. In contrast, our work opens the study of distilling different resources in sequential attempts.

 \textit{Main results}.--- We investigate the proposed concept of quantum residual management through the extraction of secondary resources from residual states left after a primary quantum information processing task. Specifically, we investigate the distillation of private randomness from residual states that remain after quantum key distribution (QKD) between two users. Both private randomness distillation and QKD represent fundamental quantum information processing tasks that are essential for secure communication against quantum adversaries~\cite{DM02,GL03,YHW19}.
 
In our work, we analyze residual management for two well-known QKD protocols: the Gottesman–Lo protocol~\cite{GL03} and  (coherent version of) Devetak–Winter protocol~\cite{DW}. 
For each case, we propose a method to extract private randomness from the corresponding residual states.
The problem we pose is not trivial, since by definition these two resources are complementary: maximal (local) private randomness is achieved for a local pure state, while maximal key is achieved on a shared entangled state, e.g., a maximally entangled one.
We derive achievable rates of private randomness distillation from the residual states left after private key distillation, using both the Gottesman–Lo (GL) and the (coherent version of) Devetak–Winter (DW) protocols. In the case of the GL protocol, we compose it with the private randomness generation protocol introduced in \cite{YHW19}. The private randomness rate from the residual of the GL protocol is non-negligible, reaching $\approx 0.114$ in the case where a non-zero key is guaranteed from the isotropic state.
Since the latter protocol is strongly related to the BBPSSW entanglement distillation protocol \cite{BBPSSW}, we also provide the gain of private randomness distilled from (part of) the residual of the latter protocol.

In case of DW protocol, our results show that all the correlations measured in terms of the quantum mutual information $I(X;B)_{\rho}$  
of a classical-quantum-quantum state $\rho_{XBE}$ (system $E$ is held by eavesdropper and $XB$ is shared by the honest parties) can be split into two independent resources: private key at usual rate $\approx [I(X;B)-I(X;E)]_{\rho}$ and private randomness at rate $\approx I(X;E)_{\rho}$, where private randomness is obtained from the residual of coherent DW key distillation protocol. Moreover, obtaining the latter destroys the shielding system of one of the honest parties. This is a benefit from a cryptographic point of view, since the honest parties need not conceal it from the adversary or spend additional resources on its physical destruction. 

Finally, we propose a formalism for the systematic study of quantum waste management. Given a noisy input resource state in some resource theory, we first distill the preferred resource, which also produces a state (to be discarded) referred to as the residual. This residual state can serve as a resource for some other resource theory, thereby allowing for an additional resource distillation protocol. We call the composition of two such protocols a sequential distillation protocol. In our framework, we restrict the sets of free states and free operations in the first resource theory to include the sets of free states and free operations of the second resource theory. 
This restriction excludes situations where free operations—applied within the secondary resource distillation—might otherwise enable the extraction of more primary resource than what could be obtained using only the free operations defined in the primary resource theory. Furthermore, consider two resource theories, $R_1$ and $R_2$, that share the same set of free operations, with the set of free states of $R_2$ contained within that of $R_1$. In such a case, it is natural to prioritize the distillation of the resource associated with $R_1$ (whenever feasible). This is because the residual state resulting from this process may still allow for subsequent distillation of $R_2$. In contrast, if one begins by distilling the resource in $R_2$, the resulting residual state generally does not permit further distillation of the resource in $R_1$. Based on this framework, we assign each resource a node and connect two nodes with a directed edge reflecting this inclusion relation. This naturally generates a directed acyclic graph (DAG) structure, which we refer to as the Residual Use Graph (RUG). The RUG captures the viable options for performing sequential distillation of multiple resources, one after another.  The proposed scenario describes a single-party resource waste management strategy; however, other approaches can naturally be taken with different restrictions~\cite{dhruvfuture}.

\textit{Illustrations}.---  We begin by presenting two simple (toy) examples that clearly illustrate our proposed idea, followed by more detailed and non-trivial cases of repurposing quantum states in sequential resource distillation tasks of practical interest. This is a warm-up for our main examples, which distill private randomness from the residual states of private key distillation using the Gottesman-Lo protocol and an adapted version of the Devetak-Winter protocol. We consider private randomness as a secondary resource given its wide importance across various applications~\cite{ACC+25,BeraReview}, for example, one-use token generation for online banking, password suggestions in online applications, and secret key generation~\cite{Gisin-crypto}. 
Private randomness is also one of the indisputable successful practical applications of quantum information theory due to quantum randomness generation protocols (see review in~\cite{Pirandola2020}). It is also related to the extractable work from a quantum state~\cite{HHH+05,DS25,BGD25}.

\textit{Notations for entropies}. The Shannon entropy of the random variable $X$ with probability distribution $P_X(x)$ is given as $H(X):=-\sum_xP(x)\log P(x)$. Without loss of generality, we consider $\log$ with base $2$. The von Neumann entropy of a quantum state $\rho_A$ is denoted by $H(A)_{\rho}=H(\rho_A):=-\tr(\rho_{A}\log\rho_{A})$. The quantum mutual information of a bipartite state $\rho_{AB}$ is defined as $I(A;B)_{\rho}:=H(A)_{\rho}-H(A|B)_{\rho}$, where $H(A|B)_{\rho}:=H(AB)_{\rho}-H(B)_{\rho}$ is the conditional von Neumann entropy of $A$ given $B$. For a classical-quantum state, $\rho_{XA}=\sum_x P(x) \op{x}_X\otimes\rho^x_{A}$, where $X$ denotes the classical register and $A$ is the quantum system, the quantum mutual information $I(X;A)_{\rho}=H(A)_{\rho}-\sum_xP(x)H(\rho^x_{A})$.

{\it Toy example 1 (Probabilistic protocol)}:
Consider a probabilistic entanglement distillation protocol~\cite{BBPS96} where a source provides two copies of two-qubit pure entangled state $\ket{\Psi}_{AB}:=a\ket{00}_{AB}+b\ket{11}_{AB}$ shared between Alice ($A$) and Bob ($B$), where $\{|0\rangle,~|1\rangle\}$ is a set of two orthonormal states. We observe that 
$\ket{\Psi}^{\otimes 2}_{AB}=a^2\ket{00}_{A_1B_1}\ket{00}_{A_2B_2}+ab(\ket{00}_{A_1B_1}\ket{11}_{A_2B_2}+\ket{11}_{A_1B_1}\ket{00}_{A_2B_2})+b^2\ket{11}_{A_1B_1}\ket{11}_{A_2B_2}$. We adopt the sequence of protocols as follows. Both Alice and Bob perform an incomplete von Neumann measurement with
projectors $P_{00}=\op{00}_{A_1A_2}$, $P_{11}=\op{11}_{A_1A_2}$, $P_{\mathrm{ent}}=[\op{01}+\op{10}]_{A_1A_2}$. The measurement outcomes are
\begin{enumerate}
\item[(a)] Output `$\mathrm{ent}$':
$\rho_{|P_{\mathrm{ent}}}=|\chi\rangle\langle\chi|$, where $|\chi\rangle=\frac{1}{\sqrt{2}}(\ket{0011}+\ket{1100})_{A_1B_1A_2B_2}
$, with probability $2a^2b^2$;
\item[(b)] Output `$00$':
$\rho_{|P_{\mathrm{random}}}=\op{0000}_{A_1A_2B_1B_2}$, with probability $a^4$;
\item[(c)] Output `$11$': $\rho_{|P_{\mathrm{work}}}=\op{1111}_{A_1A_2B_1B_2}$ with probability $b^4$.
\end{enumerate}
Let $\mathrm{CNOT}_{AB}=\sum_{i,j=0}^1|i\rangle\langle i|_{A}\otimes|i\oplus j\rangle\langle  j|_{B},$ where $\oplus$ denotes the operation sum modulo 2.
The distillation process is now based on the probabilities of outcomes, as follows:
\begin{enumerate}
\item[(a)] Given outcome `${\rm ent}$', Alice and Bob perform $\mathrm{CNOT}_{A_1A_2}$ and $\mathrm{CNOT}_{B_1B_2}$, respectively, to obtain a maximally entangled state $\frac{1}{\sqrt{2}}(\ket{00}+\ket{11})_{A_1B_1}$ and a pure state
$\ket{11}_{A_2B_2}$. The former is the desired maximally entangled pair. They transform the latter by the Hadamard gate $H$ into (the total of)
2 bits of private randomness.
\item[(b)] Given outcome `$00$', each of them apply $H^{\otimes 2}$ gates to obtain (in total) $4$ bits of private randomness.
\item[(c)] Given outcome `$11$', each of them possess $2$ bits of (local) activity (i.e., the most excited state which can be used to extract maximal work~\cite{ABN04,SSC23,CGQ+24})\footnote{We assume here, that the pure state was distributed on energetic degrees of freedom, i.e., not encoded in polarizations of photons. This is the case, e.g., when the singlet state is used in a quantum computer implemented on a platform other than a photonic one.}.
\end{enumerate}
We note that the above example illustrates the idea of managing residuals through a sequence of increasingly less available classes of operations. The distillation of entanglement needed von Neumann measurement and bilateral CNOTs.
Distilling private randomness from purity needed just a single qubit gate-- Hadamard, while distillation of the activity did not require any operation. We could also have chosen to distill private randomness or purity instead of activity. In Fig~\ref{fig:exemplary_RUGs}, we provide graphs depicting the possibility of sequential distillation of different resources starting from a source state. 


\begin{figure}[ht]
\centering
      \begin{tikzpicture}
          [edge from parent/.style={draw, -latex}]
          \node (root) [xshift = -1.5cm] {\(\ket{\Psi}^{\otimes 2}_{AB}\)}
            child { node [xshift = -0.7cm] {entanglement} child { node (a) [xshift = 0.7cm] {private randomness} } 
                }
            child { node {work} };
            \draw [-latex] (root) to node {} (a);
            \node [xshift = 1.5cm] {\(\rho_{AB}^{\otimes n}\)}
            child { node {private key} child { node {private randomness} } };
        \end{tikzpicture}
        
\caption{Pictorial representations of possible sequential resource distillation for the first toy example (on the left) obtained from processing of two copies of a pure entangled state $\ket{\Psi}_{AB}$ and the second toy example and main examples (on right) obtained from processing $n$ copies of some input mixed entangled state $\rho_{AB}$, considered in this work.}
\label{fig:exemplary_RUGs}
\end{figure}
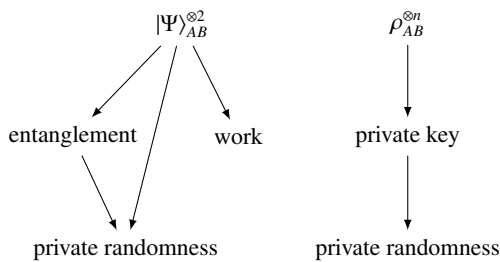

We now focus on two of the most widely studied quantum resources: private randomness and private key, as reducing residuals in these cases is highly desirable from a resource management perspective.

{\it Toy example 2 (A case of accept/abort mechanism in QKD)}: The idea of residual states can often be applied to quantum key distribution (QKD), regardless of whether it's device-dependent (trusting state source and measurement dimensions) or device-independent~\cite{PAB+09,VV19}. A fundamental aspect of all QKD protocols is their reliance on Quantum Bit Error Rate (QBER) estimation. Should the QBER be too high to generate a shared secret key—a random correlated bit string known only to the legitimate parties, Alice and Bob—the protocol is aborted. This is due to the excessive noise in the input state, which makes secret key distillation impossible. While the utilization of the output state after such an abort event has not been previously proposed to our knowledge, it's a natural consideration. Specifically, if the QBER is sufficiently low such that that the conditional min-entropy with respect of Eve is positive  ($\log|A|-H(A|B)_+>0$ in the asymptotic case), one of the parties (here Alice) can distill private randomness from the obtained data \cite{Berta2014,YHW19}. Note that we use $|A|$ for the dimension of system $A$ and $H(A|B)_+=\max\{H(A|B),0\}$.
This implies that, in certain scenarios, honest parties might be able to recover some private randomness, which is used in the key distribution phase.

We now discuss nontrivial examples and our main results.

\textit{Main Examples}. Realization of a particular path in the resource utilization leads naturally to the notion of {\it sequential protocols}. Consider first a traditional protocol ${\mathcal D}_{\mathbf{Res}_1}[S_{in}\rightarrow S_{out}]$ with input system $S_{in}$ and output $S_{out}$, distilling resource $\mathbf{Res}_1$. Based on this protocol, we define its coherent version  ${\mathcal D}[{S_{in}\rightarrow S_{\mathbf{R}_1}{G}_{\mathbf{R}_1}}]$ with the system $S_{\mathbf{R}_1}$ enabling direct use of the resource and a garbage system ${G}_{\mathbf{R}_1}$, which consist of a {\it residual} of it. By residual we mean here (i) all subsystem that would be traced out by $\mathcal D$ protocol,  or (ii) subsystem which is a part of $S_{out}$ that is not directly used when $\mathbf{Res}_1$ is utilized (such as a shielding system of a private key), but is either physically destroyed or kept untouched for security reasons. Accordingly, the exemplary composition of a sequence of two such distillation protocols acting on the input state $\rho_{in}$ we denote as follows:

 \begin{equation}\label{eq:seqdef}
 (I_{S_{\mathbf{R}_1}}\otimes{\mathcal D}^{(2)}[G_{\mathbf{R}_1}\rightarrow S_{\mathbf{R}_2}G_{\mathbf{R}_2}])\circ{\mathcal D}^{(1)}[S_{in}\rightarrow S_{\mathbf{R}_1} G_{\mathbf{R}_1}](\rho_{in})
 \end{equation}
 and call an $\mathbf{Res}_{1}\rightarrow\mathbf{Res}_2$ {\it sequential distillation protocol} (with natural generalization to more than 2 resources). Specifically we study ${\bf Key}\rightarrow{\bf PRand}$ sequential distillation protocol, meaning the ${\bf Key}$ $\rightarrow$ ${\bf Private\,Randomness}$ sequential distillation one.
\\ 
{\it Main example 1 (Two-way distillation protocol)}: There are many practical scenarios that require probabilistic protocols, as illustrated in {\it Toy example 1.} For instance, in a high-noise regime, QKD protocols often involve two-way public communication, and their outcomes are typically probabilistic. This also applies to two-way entanglement distillation protocols.  Specifically, all input noisy states $\rho_{AB}^{\otimes n}$ are divided into small subsets, such as pairs or triples. After an operation on each subset (for example, a pair), if the measurement of the second state in the pair is correlated, the first state is kept, and its key or entanglement property is improved. However, if anti-correlations are observed, the parties discard the first state of the pair. According to the philosophy presented here, such discarded outputs should be utilized for the distillation of other resources. For example, in this case, they could be used as a source of private randomness or shared public randomness. We exemplify this idea in detail via Gottesman-Lo protocol~\cite{GL03}.

{\it Description of the Gottesman-Lo protocol}:
 It is known, that there exist states which do not permit one-way key distillation. E.g. in the case of BB84 protocol it is known, that above 14.6$\%$ of quantum bit error rate (QBER), no protocol with only one-way communication can yield non-zero key secure against quantum adversary~\cite{GL03}. However, one can get positive key rate up to $18.9\%$ of QBER using two-way public communication. The Gottesman-Lo protocol can be used to distill key out of such states. It works as follows. It amounts to alternation of two steps:
the step B (bit error correction) and the step P (phase error correction) until the state reaches positive coherent information and can be distilled efficiently by one-way protocol such as Devetak-Winter one. We assume all the measurement outcomes to be classical, possibly public, and inherently beyond the paradigm of the distillation of private randomness.  We omit the classical outcomes of the involved post-selection and quantify the quantum ones. Without loss of generality we may assume that the input to the protocol are $n$ copies of some quantum state $\rho_{AB}$. We now describe it in greater detail.

Let $Z=|0\rangle\langle0|-|1\rangle\langle1| $ and $X=|0\rangle\langle1|+|1\rangle\langle 0|$. The step B amounts to (i) paring up states into $n/2$ pairs $\rho_{AB}\otimes \rho_{A'B'}$ (ii) performing bilateral XORs ($\mathrm{CNOT}_{AA'}$ and $\mathrm{CNOT}_{BB'}$) (iii) measuring systems $A'B'$ in computational basis and comparing the results. If the results are same, $\rho_{AB}$ (half of the pair $\rho_{AB}\otimes \rho_{A'B'}$) is kept for the step P, else it gets discarded. We discard the classical results from $A'B'$ after comparison is done. Notice that the step B is similar to the advantage distillation done in \cite{BBPSSW}.
In the step P, (i) parties $A$ and $B$ form trios of states ($\rho_{AB}\otimes \rho_{A'B'}\otimes \rho_{A''B''}$)  (ii)  measure observables $X_A\otimes X_{A'}$, $X_{A'}\otimes X_{A''}$, $X_B\otimes X_{B'}$, $X_{B'}\otimes X_{B''}$. It is done using Hadamard operator on all qubits, both parties applying two bilateral XORs ($\mathrm{CNOT}_{A A'}\otimes\mathrm{CNOT}_{A A''}$ and $\mathrm{CNOT}_{BB'}\otimes\mathrm{CNOT}_{BB''}$), measuring $Z_{A'}\otimes Z_{B'}$ and $Z_{A''}\otimes Z_{B''}$ and applying Hadamard operator again on the first system.
(iii) if they disagree on both measurements, they apply $Z_A\otimes Z_B$, else nothing. 

{\it Gottesmann-Lo {\bf Key}$\rightarrow$ {\bf PRand} sequential protocol}:
We now describe the proposed modification of the above protocol to extract private randomness from the residual systems after the key distillation. The idea is simple: collect all quantum states that were to be discarded in the Gottesman-Lo (GL) protocol and distill private randomness from them using \textit{private randomness distillation} protocol~\cite{YHW19}. We call such a modified protocol as Gottesmann-Lo $\mathbf{Key}\rightarrow\mathbf{PRand}$ {\it sequential distillation protocol}. 
As the first main result we obtain
the following,
\begin{result}\label{res:1}
    The private randomness rate after $r^{th}$ execution of step B in the GL protocol is given by,
\begin{align}\label{eq:gorate}
\mathrm{Rate}_{\mathrm{Key}_{\mathrm{GL}}\rightarrow \mathrm{PRand}}(r)=~~~~~~~~~~~~~~~~~~~~~~~~~~~~~~~~~~~~~~~~~~~~~~~~~~~~~~~~~~~&\nonumber \\
\begin{cases} 
      \frac{1}{2} p_{fail}(1)R^1_A, & r=1 \\
 \frac{1}{2} \left(p_{fail}(1)R^1_A +\sum_{k=2}^{r}\left(\prod_{l=1}^{k-1}\frac{1-p_{fail}(l)}{6} \right) p_{fail}(k)R^k_A\right), & r>1
\end{cases}
\end{align}
where  $p_{fail}(k)$ denotes the probability of discarding the systems which are not subjected to measurement at the $k^{th}~(1\leq k\leq r)$ step B, and $R^k_A$ denotes their asymptotic private randomness rate. And for an isotropic state, $\mathrm{Rate}_{\mathrm{Key}_{\mathrm{GL}}\rightarrow \mathrm{PRand}}$ reaches close to 0.114 where a positive key is guaranteed via GL protocol.  
\end{result}

\begin{figure}[!ht]
      \centering     \includegraphics[width=\linewidth]{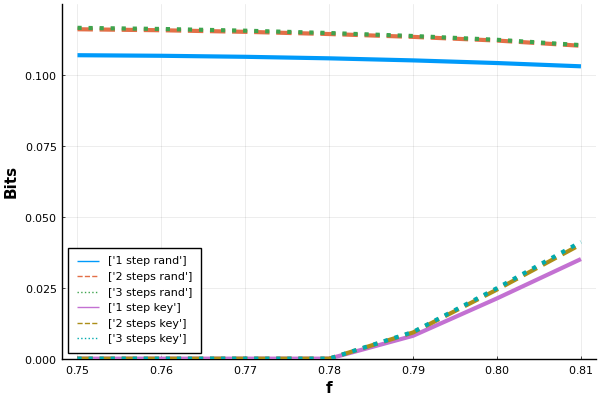}
    \caption{
    Extracted private randomness rate and lower bound on the key rate versus the maximally entangled state fraction, $f$, in the isotropic state using  Gottesman-Lo (GL) {\bf Key}$\rightarrow$ {\bf PRand} {\it sequential distillation protocol}. The `$r$ step rand' and `$r$ step key' corresponds to the $\mathrm{Rate}_{\mathrm{Key}_{\mathrm{GL}}\rightarrow \mathrm{PRand}}(r)$ and the lower bound on $\mathrm{Key_{GL}}(r)$, respectively. } 
 \label{fig:lorate}
\end{figure}

{\it Derivation of the rates of GL {\bf Key}${\rightarrow}${\bf PRand}}:
We first write the recursive formula for the rate of private randomness achieved after some steps of the alternation of the steps B and P. Let the protocol runs till the $r^{th}$ step B. Notice that at each step B half of the systems are subjected to measurement and are always discarded. Similarly, $\frac{2}{3}$ of the systems are subjected to measurement in each step P and thus are always discarded. 
Note that  $\frac{n}{2}p_{fail}(1)$ is the total number of discarded systems at the first step B, therefore, the asymptotic formula for the  randomness rate after the first step B will be $\frac{1}{n}\frac{n}{2}p_{fail}(1)R^1_A$. Then, after step  P, the total number of remaining copies are $\frac{n}{2}(1-p_{fail}(1))\frac{1}{3}$. At another step B the total number of copies being discarded is $\frac{n}{2}(1-p_{fail}(1))\frac{1}{3.2}p_{fail}(2)$ and thus the  randomness rate is $\frac{1}{2}(1-p_{fail}(1))\frac{1}{3.2}p_{fail}(2)R^2_A$. And thus the total randomness rate after second step B will be $\frac{1}{2}(p_{fail}(1)R^1_A+\frac{1-p_{fail}(1)}{6}p_{fail}(2)R^2_A)$. Similarly, we can see that after the $r^{th}$ step B the total randomness rate reads as given in Eq. \eqref{eq:gorate}. Similarly, we can write the form of the distilled key after $r^{th}$ step B, denoted by $\mathrm{Key_{GL}}(r)$. Let $K^k$ denote the asymptotic key rate from systems kept at the $k^{th}$ step B. Then,
\begin{align}\label{eq:gokeyrate}
   \mathrm{Key_{GL}}(r)=~~~~~~~~~~~~~~~~~~~~~~~~~~~~~~~~~~~~~~~~~~~~~~~~~~~~~~~~~~~~~~~~~~~~~~~~~~ \nonumber \\
   \begin{cases} \frac{1}{2}  \left(1-p_{fail}(1)\right)K^1, &r=1 \\ 
   \frac{1}{2} \left( \left(1-p_{fail}(1)\right)K^1+\sum_{k=2}^{r}\left(\prod_{l=1}^{k-1}\frac{1-p_{fail}(l)}{6} \right) K^k\right), &r>1
   \end{cases}
\end{align}
\noindent
In the Supplementary Material B, we explicitly provide the form of $\mathrm{Rate}_{\mathrm{Key}_{\mathrm{GL}}\rightarrow \mathrm{PRand}}(r)$ and  $\mathrm{Key_{GL}}(r)$ after $r^{\text{th}}$ step B for a general mixture of Bell states in the Gottesmann-Lo protocol. Next, consider the isotropic state, $\rho_{f}$, in form: 
\begin{equation}
\rho_f:=f |\psi_+\>\<\psi_+| + \frac{(1-f)}{3}\left(|\psi_-\>\<\psi_-|+ |\phi_+\>\<\phi_+|+|\phi_-\>\<\phi_-| \right),
\label{eq:iso_state}
\end{equation} 
where $|\psi_\pm\> =\frac{1}{\sqrt{2}}(|00\>\pm |11\>)$ and $|\phi_\pm\> = \frac{1}{\sqrt{2}}(|01\> \pm |01\>)$.
In the fig. \ref{fig:lorate}, we plot the key rate and private randomness rate of the isotropic state with a maximally entangled state fraction, $f$, for the modified Gottesmann-Lo $\mathbf{Key}\rightarrow\mathbf{PRand}$ {\it sequential distillation protocol}.  
As expected, $\mathrm{Rate}_{\mathrm{Key}_{\mathrm{GL}}\rightarrow \mathrm{PRand}}(r)$ increases with $r$, however, after the second step B the improvement fades away due to the contribution of the factor $\frac{1}{6^{k-1}}$ from $k^{th}$ steps B and P. 
We also observe that the Gottesmann-Lo (GL) protocol becomes more relevant for the case when $0.78\lessapprox f\lessapprox 0.81$, because in this range GL protocol guarantees a non-zero key from $\rho_{f}$, whereas, one-way (e.g. Devetak-Winter) key distillation protocol guarantees a non-zero key rate for $f\gtrapprox 0.81$ (since it is possible to get a keyrate of $-H(A|B)$ from an arbitrary state $\rho_{AB}$~\cite{DW}). Thus one-way (e.g. Devetak-Winter) key distillation protocol can be directly implemented for $f\gtrapprox 0.81$. We see that the usage of residual states allows for $\approx 0.1137$ of private randomness rate for $f=0.79$ where a non-zero key distillation is guaranteed from the GL protocol.  
We also investigate the BBPSSW entanglement distillation protocol~\cite{BBPSSW} which is same to GL protocol with the absence of steps P. The  randomness rate from the residual states of the BBPSSW protocol is plotted in Fig. \ref{fig:bbpsswrate}. Notice, since the GL protocol and the BBPSSW protocol start with the same step B, the output randomness is the same after the first step B. However, the absence of step P in the BBPSSW protocol leads to a better private randomness rate from second step B onward. Indeed, for $f=0.79$, we observe that $\approx 0.1148$ of private randomness rate can be extracted from residual of the BBPSSW protocol.

\begin{figure}
	\centering
	\includegraphics[trim={2cm 2 3 2},width=7.5cm]{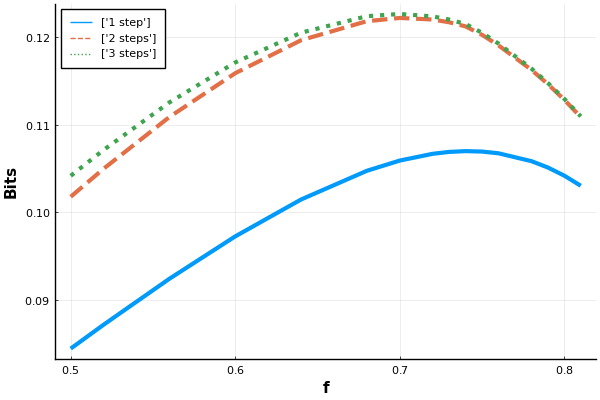}
	\caption{
    Randomness rate from the residual states of the BBPSSW entanglement distillation protocol  versus the maximally entangled state fraction, obtained from the isotropic state given in Eq.(\ref{eq:iso_state}). Here `$r$ steps' indicate the curve for the randomness rate after $r^{th}$ step B.}
 \label{fig:bbpsswrate}
\end{figure}

{\it Main example 2  (Private key followed private randomness generation)}: The Devetak–Winter protocol \cite{DW} for quantum key distribution has, for the last two decades, been one of the most commonly used protocols for distilling a secret key from quantum states assisted by one-way communication (see, however,~\cite{RenesRenner} for an alternative one-shot approach). We will consider its coherent version, i.e., one in which all local operations are implemented by (i) adding auxiliary systems (ii) performing unitary transformation (iii) put aside any system instead of tracing it out (for a formal description of such protocols, see \cite{HHHO09}).
Such a protocol explicitly generates not only the subsystem used to obtain the key but also a residual subsystem, known as the shield, which should be kept away from the eavesdropper \cite{pptkey}. This shielding system can therefore be regarded as the residual part of the protocol. We will show that, for a wide class of states, one can extract private randomness directly from such obtained system. Specifically, consider pure tripartite states $\ket{\psi}_{ABE}$ which, when measured in a full von Neumann basis ${\ket{x}}$ on Alice’s (sender’s) subsystem, induce a pure ensemble ${\ket{\phi^x}_{BE}}$. We show that, given a sufficient number of copies of such a state, one can execute the coherent Devetak–Winter (DW) protocol to first extract a secret key and subsequently distill private randomness from the corresponding shield. To establish this result, we introduce a sequential DW $\mathbf{Key}\rightarrow\mathbf{Prand}$ distillation protocol (see Supplementary Material D).

\begin{result}\label{res:2}
    Consider a bipartite state $\rho_{AB}$ such that its purification to the eavesdropper's system $E$ reads
\begin{equation}
|\psi\>_{ABE}:=\sum_{x} \sqrt{P(x)}\ket{x}_A\otimes \ket{\phi^x}_{BE}.
\label{eq:input_state}
\end{equation}
Let $\sigma_{XBE}$ be the state $\ket{\psi}_{ABE}$ measured in computational basis $\{\ket{x}\}$ on system $A$.
For a sufficiently large number of copies $n$ of $\rho_{AB}$,
one can distill $r_{\mathrm{key}}\approx n [I(X;B)-I(X;E)]_{{\sigma}}$ bits of secure key by a (one-way from A to B) coherent Devetak-Winter protocol. Moreover from the residual
of the latter key distillation protocol, one can obtain $r_{\mathrm{rand}}\approx n I(X;E)_{\sigma}$ bits of randomness private from Eve. 
The latter protocol
destroys at the same time
the quantum  subsystem of X, which is shielding the key system from quantum adversary.
\end{result}

\begin{figure}
	\centering
	\includegraphics[trim={2cm 2 3 2},width=7cm]{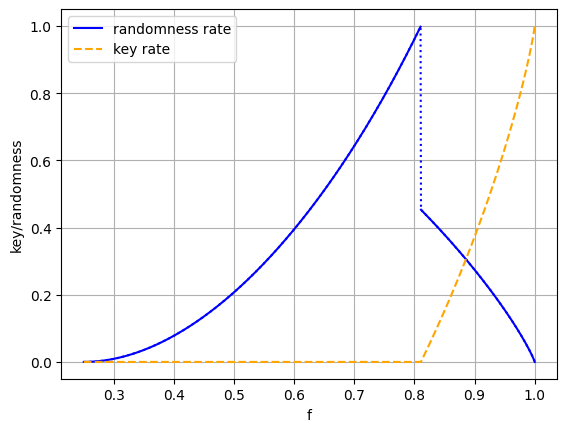}
	\caption{The figure depicts the relation between the  rate of private randomness and secret key distilled from the isotropic state (depending on parameter $f \in [0,1]$ as in Eq. (\ref{eq:iso_state})). The blue, solid curve, representing private randomness, consists of two distinct parts with a jump at $f \approx 0.8125$, where the rate of key distillation (orange, dashed curve) $-H(A|B)$  becomes positive. When $-H(A|B)$ is positive, randomness is distilled at rate $I(X;E)_{\sigma^f_{XE}}$ as shown in Result 2 and else (if $H(A|B) > 0$) it is distilled at rate $R_{AB}= \log|A| - H(A|B)_+ = \log|A| - H(A|B)$, which in our case of isotropic state is $2-H(AB)$.
    }
 \label{fig:iso_gain}
\end{figure}

The detailed proof of the Result~\ref{res:2} can be found
in the Supplementary Material D. We now provide the idea used in the sequential DW $\mathbf{Key}\rightarrow\mathbf{Prand}$ distillation protocol that generates the pair of rates $(r_{key},r_{rand})$ mentioned above. \\
{\it Idea of the proof of Result~\ref{res:2}}. We first focus on the coherent version of DW key-distillation protocol. This protocol results in the splitting of the system $A(B)$ into subsystems $\mathrm{Key}_A(\mathrm{Key}_B)$ and $\mathrm{Shield}_A(\mathrm{Shield}_B)$ such that $\mathrm{Key}_A(\mathrm{Key}_B)$ will hold key for subsystem $A(B)$ and $\mathrm{Shield}_A(\mathrm{Shield}_B)$  acts as a shield for subsystem $A(B)$. After Alice's measurement in the basis of $\{|x\rangle\}$,  she communicates a public classical message $\mathrm{Mes}_C(\mathrm{M}_C)$ to Bob. If $|\psi\rangle_{ABE}$ is good enough, this communication helps to correlate the subsystem $A$ with subsystem $B$ such that if subsystems $\mathrm{Key}_A$ and $\mathrm{Key}_B$ are measured in a computational basis they they would hold a secret key. The schematic for this protocol is depicted in Fig. 5a. In the case of the DW entanglement distillation (hashing) protocol (depicted in Fig. 5b),
after the communication $\mathrm{M}_C$,   Alice additionally measures $\mathrm{Shield}_A$ in a Fourier transform of its computational basis and communicates classical message, $\mathrm{Mes}_D(\mathrm{M}_D)$, to Bob publicly. The message $\mathrm{M}_D$ allows Bob to decouple subsystems $\mathrm{Shield}_B E$ from $\mathrm{Key}_B$. Therefore, given the protocol succeeds, the subsystems $\mathrm{Key}_A$ and $\mathrm{Key}_B$ are pure, entangled (in a nearly maximally entangled state) and hence decoupled from Eve. We thus indicate $\mathrm{Key}_A$ and $\mathrm{Key}_B$ by $\mathrm{ENT}_A$ and $\mathrm{ENT}_B$, respectively in Fig. 5b. 
Now in the sequential DW $\mathbf{Key}\rightarrow\mathbf{Prand}$ distillation protocol (see Fig. 5c), the idea is that Alice {\it does not} communicate $\mathrm{M}_D$ to Bob and keep it to herself. This turns out to be her randomness private from 
$E$. Now, Bob cannot
decouple $\mathrm{Shield}_B$ from $\mathrm{Key}_B$ locally (due to lack of information $\mathrm{M}_D$) but systems $\mathrm{Key}_A\mathrm{Key}_B$ after
measuring in the computational basis (as in the coherent DW key distillation protocol) holds secure key which, as we show, is also independent of $M_D$ and thus independent of the private randomness.  The rate of key is the same as in the coherent DW key distillation protocol $r_{key}=[I(X;B)-I(X;E)]_{\sigma}$ and the rate of private randomness is equal to the size of the shielding system of Alice, which reads $r_{rand}=I(X;E)_{\sigma}$. Note that the rate of key distillation $r_{key}=-H(A|B)_\rho$~\cite{DW}. It can be easily obtained from the fact that the state $\sigma$ is obtained after measuring $|\psi\rangle_{ABE}$ (purification of $\rho_{AB}$) in the Schmidtt basis $|x\>_A$, thus $H(BE)_\sigma=H(BE)_\psi$.
Therefore, the rate of key distillation,
$I(X;B)-I(X;E)]_{\sigma}=H(B)_\psi-\sum_xP(x)H(B)_{\phi^x}-H(E)_\psi+\sum_xP(x)H(E)_{\phi^x}=-H(A|B)_\psi=-H(A|B)_\rho$, using the facts that $H(B)_{\phi^x}=H(E)_{\phi^x}$ and $H(E)_\psi=H(AB)_\psi$.

\begin{figure}[h!]
    \includegraphics[trim={0cm 2 3 2},width=9.5cm]{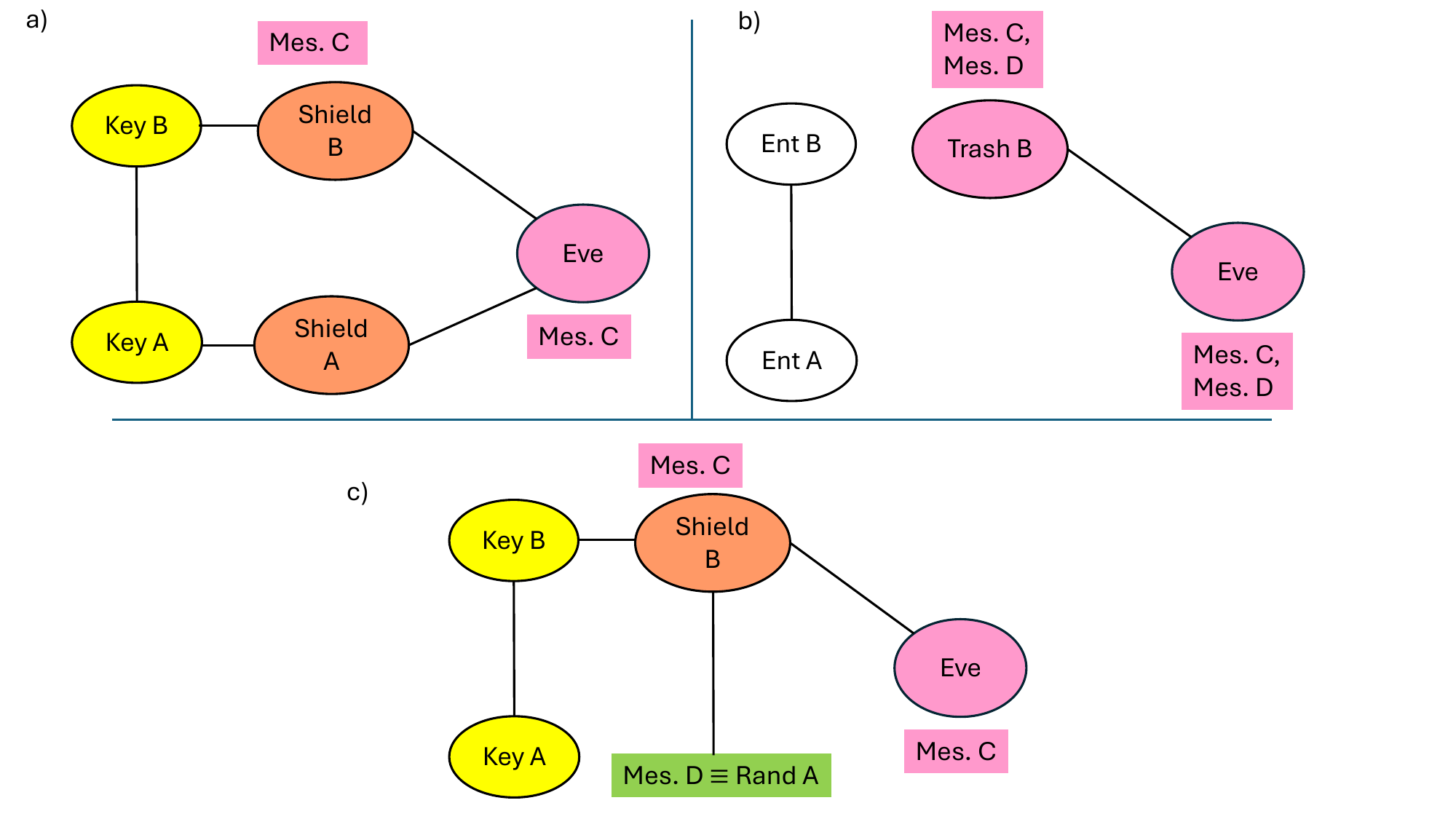}
	\caption{
    Depiction of the output of 3 protocols:  a) coherent Devetak-Winter key distillation protocol outputs a state with key parts and shielding systems as well as public message $\mathrm{M}_C$ which enabled Alice and Bob become fully correlated. b) Devetak-Winter hashing entanglement distillation protocol is the coherent DW protocol which outputs maximal entanglement on systems $\mathrm{Ent ~A}$ and $\mathrm{Ent~B}$ thanks to generating additional public message, $\mathrm{M}_D$, from Alice to Bob. 
    c) Devetak-Winter $\mathbf{Key}\rightarrow\mathbf{Prand}$ sequential protocol distills key on systems $\mathrm{Key}_A$ and $\mathrm{Key}_B$. The kept private message $\mathrm{M}_D$, which coincides with result of the measurement of the system $\mathrm{Shield~A}$ of the coherent DW protocol is shown to be private randomness for Alice, independent from the generated key, and most importantly private from the Eve.  }
\label{fig:DWidea}
    \end{figure}

In Fig. \ref{fig:iso_gain}, we plot the rate of key distillation and the private randomness rate versus the maximally entangled state fraction, $f$, for the isotropic state $\rho_f$, given in Eq. \eqref{eq:iso_state}. The key rate becomes positive for $f\gtrapprox0.8125$. The randomness rate is calculated in two parts. Until $f\approx 0.8125$, the randomness rate is given by $\log|A|-H(A|B)_+$~\cite{YHW19}. However, for positive key rate regime, the DW sequential  protocol dictates the private randomness to be $I(X;E)_{\sigma^f}$, where $\sigma^f_{XBE}$ is the classical quantum state obtained after measuring subsystem $A$ of the state $\rho_f$ in the computational basis $\{|x\rangle\}$.  

{\it Remark 1}. Proving the Result~\ref{res:2}, We show that there exist a  $\mathrm{Key}\rightarrow\mathrm{PRand}$ sequential protocol which gives additional cryptographic benefit. This is because the residual state of the so called privacy amplification step, is a part, that should be kept secret from Eve for security reasons.
Indeed, in language of the information theory it is a {\it shielding} state which protects the key from the quantum eavesdropper (see  \cite{pptkey} for the formally considered shielding system in the case of quantum adversary).
For this reason, so far one has faced a non-zero energy expenditure: either energy of  shielding system is locked from Alice for a long time (ideally forever) which clearly costs energy and storage resources or
it needs to be physically destroyed. Any such erasure has a non-zero energy cost due to famous Landauer's principle \cite{Landauer1961}. As we will show in what follows, the idea
of reuse of the shielding system for private randomness generation,
solves both
the problem of protecting the shield from eavesdropper and (to some extent) the need of private randomness in future QKD protocols.

{\it Remark 2}.  Striking as it is, from the theoretical perspective, the main Result~\ref{res:2} implies that all the correlations measured in terms of the quantum mutual information $I(X;B)_{\sigma^f}$,
can be split into two {\it independent} resources: private key and private randomness.

 We now present a general framework for quantum residual management using sequential resource distillation of higher level.
 
 \textit{Formalism}.--- In general, consider $\mathbf{Res}_i:=(\mathbf{F}_i,\mathbf{O}_i)$ be a distinct resource theory of the resource $\mathrm{Res}_i$ for each $i\in\{1,2,\ldots,r\}$, $r\in\mathbb{N}$, where $\mathbf{F}_i$ is the set of free states in $\mathbf{Res}_i$ and $\mathbf{O}_i$ is the set of free operations in $\mathbf{Res}_i$. We can form several hierarchies among $\mathbf{Res}_i$ by evaluating their free states: whether $\mathbf{F}_{i+1}\subseteq \mathbf{F}_i$ and $\mathbf{O}_{i+1}\subseteq \mathbf{O}_i$ or not? We can build a \textit{residual use graph} (RUG) $\mathbf{G}_{\mathrm{Res}_{i_0}}$ for $\mathrm{Res}_{i_0}$ by considering different resource theories (need not be exhaustive), for which their set of free states and free operations are contained in the set $\mathbf{F}_{i_0}$ and $\mathbf{O}_{i_0}$, respectively\footnote{Note that the RUG is in fact a Hesse diagram because of the set inclusion of free states and free operations gives rise to a partial order in the set of resource theories.}. Each edge of the graph is of the form $\mathbf{Res}_i\to \mathbf{Res}_{j}$, given that $\mathbf{F}_i\supseteq \mathbf{F}_{j}$ and $\mathbf{O}_i\supseteq \mathbf{O}_{j}$. Naturally, these edges constitute a transitive relation between resource theories, i.e. if there is an edge $\mathbf{Res}_i\to \mathbf{Res}_{j}$ and edge $\mathbf{Res}_j\to \mathbf{Res}_{k}$, then there is also an edge $\mathbf{Res}_i\to \mathbf{Res}_{k}$. Thus, paths connecting different resource theories with $\mathbf{Res}_{i_0}$ are not unique and may have different lengths. Nevertheless, we can introduce the notion of {\it level of the vertex} in our graph. For a resource theory $\mathbf{Res}_{j}$, we define it as the length $\ell$ of the longest path connecting $\mathbf{Res}_{i_0}$ and $\mathbf{Res}_{j}$. This allows us to a assign a natural index $\ell$ to the resource theory $\mathbf{Res}_{j}$. Thus, we will denote $\mathbf{Res}_{j}$ as $\mathbf{Res}_{i_0}[n, \ell]$, where $n$ is an arbitrary integer, unique for each resource theory of level $\ell$. The main objects of investigation in our work are paths of RUG. Each path with `$k$' nodes starting from the root $\mathbf{Res}_{i_0}$ is of the form $\mathbf{Res}_{i_0}\to \mathbf{Res}_{i_0}{[n_1,\ell_1]}\to \mathbf{Res}_{i_0}{[n_2,\ell_2]}\to\cdots\to \mathbf{Res}_{i_0}{[n_{k-1},\ell_{k-1}]}$ (see Fig.~\ref{fig:RUG}), given that $\mathbf{F}_{i_0}\supseteq \mathbf{F}_{i_0}{[n_1,\ell_1]}\supseteq \mathbf{F}_{0}{[n_2,\ell_2]}\supseteq\cdots\supseteq \mathbf{F}_{i_0}{[n_{k-1},\ell_{k-1}]}$ 
 and $\mathbf{O}_{i_0}\supseteq \mathbf{O}_{i_0}{[n_1,\ell_1]}\supseteq \mathbf{O}_{0}{[n_2,\ell_2]}\supseteq\cdots\supseteq \mathbf{O}_{i_0}{[n_{k-1},\ell_{k-1}]}$.

 By examining the structure of the RUG in greater depth, we find that it forms a directed acyclic graph (DAG). Indeed, if there was a directed cycle, resource theories representing the nodes would have {\it equal} sets of free states and operations, hence would be equal.
  It is however also instructive to associate with RUG, its undirected corresponding graph denoted here as URUG. The later is obtained from RUG by erasing direction of edges. 
 Interestingly, URUG is a so called {\it chordal graph}~\footnote{A chord is of a cycle is an edge which connects its non-consecutive vertices. A graph is called chordal if every cycle of size larger than $4$ has a chord.
},  which follows directly from the 
 transitivity of the set inclusion relation. 
 In particular, in Fig.~(\ref{fig:RUG}) there should be an edge between $\mathbf{Res}_i$ and $\mathbf{Res}_i[2, 2]$ and many others, but we omitted them to increase readability.
 
 \begin{figure}[htpb]
  \centering
  \begin{tikzpicture}[
    edge from parent/.style={draw, -latex},
    level distance=20mm,     
    sibling distance=25mm,   
    every node/.style={align=center} 
  ]
  \node {\(\mathbf{Res}_{i}\)}
    child { node {\(\mathbf{Res}_{i}{[1,1]}\)} 
        child { node {\(\mathbf{Res}_{i}{[1,2]}\)} }
        child { node {\(\mathbf{Res}_{i}{[2,2]}\)} 
            child { node {\(\mathbf{Res}_{i}{[1,3]}\)} }
        } 
    }
    child { node {\(\mathbf{Res}_{i}{[2,1]}\)} 
      child { node {\(\mathbf{Res}_{i}{[2,2]}\)} }
      child { node {\(\mathbf{Res}_{i}{[3,2]}\)} }
    }
    child { node {\(\mathbf{Res}_{i}{[3,1]}\)} };
  \end{tikzpicture}
  \caption{Pictorial representation of a Residual Use Graph (RUG) for $\mathbf{Res}_i$. 
  Nodes represent resources. Each path from the root to the leaf corresponds to 
  the sequence of resources satisfying the inclusion property: if the resource 
  $\mathbf{Res}_i$ is a parent of $\mathbf{Res}_j$ in the RUG, then the set of 
  free states of $\mathbf{Res}_j$ is contained in the set of the free states 
  of $\mathbf{Res}_i$. To increase readability of the graph, we left only the edges that connect vertices of adjacent levels.}
  \label{fig:RUG}
\end{figure}
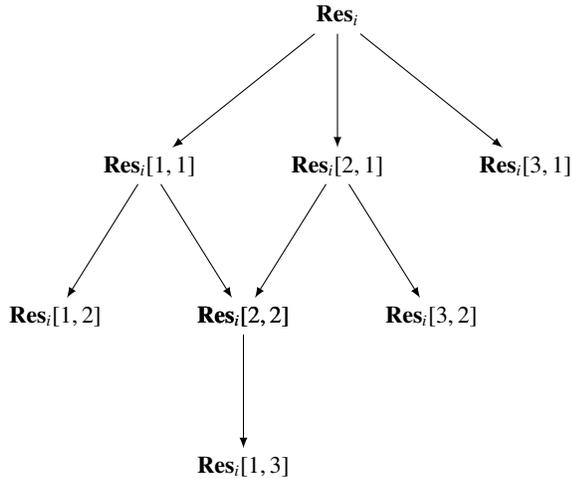
 
 For the task of quantum residual management, we can pick some nodes in a fixed order from a path (say $\mathbf{Res}_{i_0}\to \mathbf{Res}_{i_0}{[n_1,\ell_1]}\to \mathbf{Res}_{i_0}{[n_2,\ell_2]}\to\cdots\to \mathbf{Res}_{i_0}{[n_{k-1},\ell_{k-1}]}$) of a RUG (say $\mathbf{G}_{\mathrm{Res}_{i_0}}$) for $\mathbf{Res}_{i_0}$ to distill different resources in succession from residual obtained from the distillation of the respective primary resources. 
 For the residual management from the resource distillation task of $\mathbf{Res}_{i_0}$ using 
 path $\mathbf{Res}_{i_0}\to \mathbf{Res}_{i_0}{[n_1,\ell_1]}\to \mathbf{Res}_{i_0}{[n_2,\ell_2]}\to\cdots\to \mathbf{Res}_{i_0}{[n_{k-1},\ell_{k-1}]}$, we allow the distillation of resource $\mathrm{Res}_{i_0}{[n_j,\ell_j]}$ only through the allowed free operations in $\mathrm{Res}_{i_0}[n_{j},\ell_{j}]$. This is because the operation employed in the distillation protocol should not allow to increase the resourcefulness of the source quantum states.

 Lastly, the choices of actual realization of a path $\mathrm{Res}_{i_0}$ using 
 path $\mathbf{Res}_{i_0}\to \mathbf{Res}_{i_0}{[n_1,\ell_1]}\to \mathbf{Res}_{i_0}{[n_2,\ell_2]}\to\cdots\to \mathbf{Res}_{i_0}{[n_{k-1},\ell_{k-1}]}$ depend on the interests of the user among several physically feasible choices, driven by multiple factors like state source, market demands, energy cost~\cite{Auffves2022,HWS+25}, technological capabilities, etc.
In the case of {\it Main example} 1 and 2 we have considered 
 $\mathbf{Res}_{1}$ to be
 the resource theory of private key $\mathbf{Res}_1\equiv\mathbf{Key}=(\mathrm{SEP},\mathrm{LOPC})$.
Here SEP is the set of separable (i.e. disentangled) states \cite{HHHH09} and LOPC stands for the set of Local Operations and Public Communication operations \cite{DW,HHHO09,pptkey}
. We note here, that
the set of free states for the resource theory of private key secure against quantum adversary, is not known (see the 24th open problem in \cite{Viennalist} and \cite{HSDW24} in this context). It is however known, that it contains separable states \cite{CLL04}. Further we have considered resource theory of private randomness in a bipartite setting of \cite{YHW19} hence the theory of our secondary resource is $\mathbf{Res}_{1}[1,1]\equiv \mathbf{Prand}=(\mathrm{MMS}_A,\mathrm{LU_A})$, where $\mathrm{LU_A}$ stands for local unitary transformations at A's side, and the set of free states $\mathrm{MMS}_A$ consists of maximally mixed states of arbitrary but finite dimension. 
Note that indeed according to the rules considered, $
\mathbf{Res}_1
\rightarrow\mathbf{Res}_1[1,1] $, i.e., $\mathbf{Key}
\rightarrow
\mathbf{Prand}$, since a maximally mixed state is separable, so $\mathrm{MMS}_A\in \mathrm{SEP}$
 and the local unitaries on both sides are special $\mathrm{LOPC}$ operations.
 
{\it Discussion}.--- We have proposed a paradigm for resource-aware fundamental quantum resource extraction, which paves a way for more efficient quantum information processing. A systematic study of the concept of residual utility graphs for commonly used quantum resources, should be done in the first place. 

We have focused on showing achievable rates in Gottesmann-Lo and Devetak-Winter ${\bf Key}\rightarrow{\bf Prand}$ sequential protocols. The next step would be to make these protocols practically useful. In particular, it remains an interesting open problem whether the DW protocol must be performed coherently for private randomness to be extractable from its shield.

We show that the chosen distillation protocols possess useful residuals from which another resource can be extracted. However, it remains an open question whether optimal protocols (e.g., those for secure key distillation) also exhibit this property. Indeed, the key rate for the DW protocol that we obtain is equal to the distillable entanglement of the input state. It is, however, known that the secure key rate can exceed the entanglement distillation rate \cite{pptkey}. It is also important to investigate whether the DW-based sequential protocol corresponds to the so-called virtual quantum state merging protocol~\cite{Horodecki2005}, as observed in \cite{YHW19}, since we also demonstrate that private randomness can be drawn from information that would otherwise be communicated publicly but is instead used as a private resource.

Furthermore, in the spirit of the Quantum Energy Initiative, it will be useful to study the same protocols from an energy-consumption perspective in light of recent results \cite{FellousAsiani2023,Raja_com,HWS+25,BGD25b}. It should be verified whether the energy spent on drawing resources from already generated quantum data (residual of some protocol) is not greater than the energy spent on generating them from scratch.

\textit{Note added}.--- During the final preparation of this manuscript for submission to arXiv today, we noticed that a work~\cite{ZFL+25} with a similar motivation also appeared on arXiv on the same day.

\begin{acknowledgments}
C.S. and L.S. acknowledge
support of the National Science Centre, Poland, under grant Opus 25 no. UMO-2023/49/B/ST2/02468.
KH acknowledges the Fulbright Program and Cornell ECE for hospitality during the Fulbright scholarship at the School of Electrical and Computer Engineering of Cornell University.
KH acknowledges partial support by the Foundation for Polish Science (IRAP project, ICTQT, contract no.\ MAB/2018/5, co-financed by EU within Smart Growth Operational Programme). The `International Centre for Theory of Quantum Technologies' project (contract no.\ MAB/2018/5) is carried out within the International Research Agendas Programme of the Foundation for Polish Science co-financed by the European Union from the funds of the Smart Growth Operational Programme, axis IV: Increasing the research potential (Measure 4.3).
K.H. acknowledges partial support
of the National Science Centre, Poland, under grant Opus 25,
2023/49/B/ST2/02468.
S.D. acknowledges support from the SERB (now ANRF), Department of Science and Technology, Government of India, under grant no. SRG/2023/000217 and the Ministry of Electronics and Information Technology (MeitY), Government of India, under grant no. 4(3)/2024-ITEA, and the National Science Centre, Poland, under grant Opus 25, 2023/49/B/ST2/02468. S.D. thanks University of Gda\'nsk, Gda\'nsk, Poland for the hospitality during his visit. The authors thank Dhruv Baheti for spotting an error in the early version of the manuscript.
\end{acknowledgments}

\section*{Supplemental Material}
\subsection{Preliminaries on private randomness}\label{supp0}

For a given state $\rho_{AB}$, the following private randomness rate regions are achievable (and tight in settings 1, 2 and 3)~\cite{YHW19}:\\
1) For no communication and no noise, $R_A \leq \log |A|-H(A|B)_+,$ $R_B \leq \log |B|-H(B|A)_+,$ and $R_A +R_B \leq R_G,$ where $[t]_+ = \max\{0, t\}$;\\
2) for free noise but no communication, $R_A \leq \log |A|-
H(A|B)$, $R_B \leq \log |B| - H(B|A)$, and $R_A + R_B \leq R_G$;\\
3) for free noise and free communication, $R_A \leq R_G,$ $R_B \leq R_G,$ and $R_A + R_B \leq R_G$;
4) for free communication but no noise, $R_A \leq
\log |AB| - \max\{H(B), H(AB)\}$, $R_B \leq \log |AB| -
\max\{H(A), H(AB)\}$, and $R_A + R_B \leq R_G$.

\subsection{Private randomness in Gottesmann-Lo key distillation protocol}\label{supp1}
Consider the initial state of system $A$ and $B$ be 
\begin{equation}
    \rho^0_{AB}=w_0 \psi_+ +x_0 \psi_-+y_0\phi_++z_0\phi_-,
\end{equation}
where $w_0,x_0,y_0,z_0\geq 0$ and $w_0+x_0+y_0+z_0=1$. As already mentioned in the main text, this protocol has alternating step B and step P. We also assumed that we stop after the $r^{th}$ step B, where $r$ is some natural number. 
After the first step B, the state of the discarded system is given by
\begin{equation}    \rho^{fail(1)}_{AB}=\frac{(w_0y_0+x_0z_0)(\psi_++\phi_+)+(w_0z_0+x_0y_0)(\psi_-+\phi_-)}{p_{fail}(1)}
\end{equation}
where $p_{fail}(1)=2(w_0+x_0)(y_0+z_0)$ is the probability of discarding the state in the first step B.  
For an arbitrary state, $\rho_{AB}$, let the maximum private randomness extractable from system $A$, using no communication and affected by no noise, be denoted by $R_A(\rho_{AB})$,  then from \cite{YHW19}, we have
\begin{equation}
    R_A(\rho_{AB})=\log|A|-H(A|B)_+
\end{equation}
where $H(A|B)_+=\max\left\{H(A|B),0 \right\}$ and $H(A|B)=H(\rho_{AB})-H(\rho_B)$. Let $R^k_A$ denote the asymptotic private randomness rate after $k^{th}$ step B, where $1\leq k\leq r$. So $R^1_A=R_A(\rho^{fail(1)}_{AB})$.
Next, it is easy to see that the accepted systems after the first B step are in the state given by
\begin{align}
   \rho^{B1}_{AB}&= a_1\psi_+ 
   + b_1\psi_-  + c_1\phi_+ + d_1\phi_-,
\end{align}
where $a_1=\frac{w_0^2+x^2_0}{(w_0+x_0)^2+(y_0+z_0)^2}$, $b_1=\frac{2w_0x_0}{(w_0+x_0)^2+(y_0+z_0)^2}$, $c_1=\frac{y_0^2+z_0^2}{(w_0+x_0)^2+(y_0+z_0)^2}$, and $d_1=\frac{2y_0z_0}{(w_0+x_0)^2+(y_0+z_0)^2}$.
There exists a protocol such that the asymptotic key rate extracted from an arbitrary state $\rho_{AB}$ is given by the coherent information, $-H(A|B)_{\rho_{AB}}$~\cite{DW}. Let us denote the asymptotic key rate generated using such a protocol from the accepted systems after $k^{th}$ step B by $K^k$, where $1\leq k\leq r$. Then, $K^1=-H(A|B)_{\rho^{B1}_{AB}}$.
Now after applying the first  step P, the state of the system $A$ and $B$ can be in any of the following two states,
\begin{align}
   \rho^{P1}_{AB}= w_1 \psi_+ + x_1 \psi_-+ y_1\phi_+ + z_1\phi_-
\end{align}
\begin{align}
   \rho^{'P1}_{AB}= x_1 \psi_+ + w_1 \psi_- + z_1\phi_+ + y_1\phi_-,
\end{align}
with probability $q_1:=1-(a_1^2+c_1^2)(b_1+d_1)-(b_1^2+d_1^2)(a_1+c_1)-2a_1b_1(c_1+d_1)-2c_1d_1(a_1+b_1)$ and $1-q_1$, respectively, and where 
\begin{align}
w_1&=a_1^3+2a_1^2b_1+a_1b^2_1+3a_1c_1^2 \nonumber \\
&+2b_1c_1^2 +a_1d_1^2+4a_1c_1d_1+2b_1c_1d_1, \nonumber \\
x_1&=b_1^3+a_1^2b_1+2a_1b^2_1+3b_1d_1^2\nonumber \\
&+b_1c_1^2+2a_1d_1^2+4b_1c_1d_1+2a_1c_1d_1, \nonumber \\
y_1&=c_1^3+2c_1^2d_1+c_1d^2_1+3a_1^2c_1 \nonumber \\
&+2a_1^2d_1+b_1^2c_1+4a_1b_1c_1+2a_1b_1d_1, \nonumber \\
z_1&=d_1^3+c^2_1d_1+2c_1d_1^2+3b_1^2d_1 \nonumber \\
&+2b^2_1c_1+a_1^2d_1+4a_1b_1d_1+2a_1b_1c_1. 
\end{align}
However, any of these states, when passed through next step B, will give the state of the discarded system as
 \begin{equation}
    \rho_{AB}^{fail(2)}=\frac{(w_1y_1+x_1z_1)(\psi_++\phi_+)+(w_1z_1+x_1y_1)(\psi_-+\phi_-)}{p_{fail}(2)}
\end{equation}   
 with the probability $p_{fail}(2)=2(w_1+x_1)(y_1+z_1)$. And $R^2_A=R(\rho_{AB}^{fail(2)})$. 
Let us define $\theta_1=(w_1+x_1)(y_1+z_1)$ and $\theta_0=(w_0+x_0)(y_0+z_0)$. We can see that
\begin{equation}
   \theta_1= \frac{\theta_0^2[3((1-2\theta_0)^2-2\theta^2_0)^2+4\theta^4_0]}{(1-2\theta^2_0)^6}. 
\end{equation}
Now, the state of the accepted systems after the second B step is given by
\begin{equation}
    \rho^{B2}_{AB}=a_2\psi_+ 
   + b_2\psi_-  + c_2\phi_+ + d_2\phi_-,
\end{equation}
 where $a_2=\frac{w_1^2+x^2_1}{(w_1+x_1)^2+(y_1+z_1)^2}$, $b_2=\frac{2w_1x_1}{(w_1+x_1)^2+(y_1+z_1)^2}$, $c_2=\frac{y_1^2+z_1^2}{(w_1+x_1)^2+(y_1+z_1)^2}$, and $d_2=\frac{2y_1z_1}{(w_1+x_1)^2+(y_1+z_1)^2}$. Notice that $\rho^{B2}_{AB}$ is same whether that state after first P step is $\rho^{P1}_{AB}$ or $\rho^{'P1}_{AB}$. So $K^2=-H(A|B)_{\rho^{B1}_{AB}}$.

Then, using mathematical induction, the state of the system after the $k-1^{\text{th}}$ P step will be of the form $w_{k-1} \psi_+ + x_{k-1}\psi_- + y_{k-1}\phi_+ + z_{k-1}\phi_-$.  And after the $k^{\text{th}}$ B step the state of the accepted system will be $a_{k} \psi_+ + b_{k}\psi_- + c_{k}\phi_+ + d_{k}\phi_-$, where $a_k=\frac{w_{k-1}^2+x^2_{k-1}}{N_{k-1}}$, $b_2=\frac{2w_{k-1}x_{k-1}}{N_{k-1}}$, $c_2=\frac{y_{k-1}^2+z_{k-1}^2}{N_{k-1}}$, $d_2=\frac{2y_{k-1}z_{k-1}}{N_{k-1}}$, and $N_{k-1}=(w_{k-1}+x_{k-1})^2+(y_{k-1}+z_{k-1})^2$.  
Then after the
$k^{\text{th}}$ P step it will become $w_{k} \psi_+ + x_{k}\psi_- + y_{k}\phi_+ + z_{k}\phi_-$ or $x_{k} \psi_+ + w_{k}\psi_- + z_{k}\phi_+ + y_{k}\phi_-$, where 
\begin{align}
w_k&=a_k^3+2a_k^2b_k+a_kb^2_k+3a_kc_k^2 \nonumber \\
&+2b_kc_k^2 +a_kd_k^2+4a_kc_kd_k+2b_kc_kd_k, \nonumber \\
x_k&=b_k^3+a_k^2b_k+2a_kb^2_k+3b_kd_k^2\nonumber \\
&+b_kc_k^2+2a_kd_k^2+4b_kc_kd_k+2a_kc_kd_k, \nonumber \\
y_k&=c_k^3+2c_k^2d_k+c_kd^2_k+3a_k^2c_k \nonumber \\
&+2a_k^2d_k+b_k^2c_k+4a_kb_kc_k+2a_kb_kd_k, \nonumber \\
z_k&=d_k^3+c^2_kd_k+2c_kd_k^2+3b_k^2d_k \nonumber \\
&+2b^2_kc_k+a_k^2d_k+4a_kb_kd_k+2a_kb_kc_k. 
\end{align}

\noindent
Thus, the state of the discarded system after $k+1^{th}$ B step is given by
 \begin{equation}
    \rho_{AB}^{fail(k+1)}=\frac{(w_ky_k+x_kz_k)(\psi_++\phi_+)+(w_kz_k+x_ky_k)(\psi_-+\phi_-)}{p_{fail}(k+1)}
\end{equation}   
 with the probability $p_{fail}(k+1)=2(w_k+x_k)(y_k+z_k):=2\theta_k$. The sequence $\theta_k$ follows the recursion relation:
\begin{equation}
    \theta_k=\frac{\theta_{k-1}^2[3((1-2\theta_{k-1})^2-2\theta^2_{k-1})^2+4\theta^4_{k-1}]}{(1-2\theta^2_{k-1})^6},
\end{equation} 
Now the rate of private randomness after the $r^{th}$ step B can be obtained using Eq. \eqref{eq:gorate}. That is,
\begin{align}\label{pawn} 
\mathrm{Rate}_{\mathrm{Key}_{\mathrm{GL}}\rightarrow \mathrm{PRand}}(r)=~~~~~~~~~~~~~~~~~~~~~~~~~~~~~~~~~~~~\nonumber \\
\begin{cases} 
      \theta_0 R^1_A, & r=1 \\
     \theta_0 R^1_A+\sum_{k=1}^{r-1}\left(\prod_{l=0}^{k-1}\frac{1-2\theta_l}{6} \right) \theta_kR^{k+1}_A, & r>1. 
   \end{cases}
\end{align}

Similarly, we have $K^k$ for $1\leq k\leq r$, the key rate after $r^{th}$ step B can be obtained using  Eq.~\eqref{eq:gokeyrate}.

\subsection{Definitions}\label{supp5}

For the sake of completeness, here we present definitions on $\epsilon$-secrecy, $\epsilon$-evenness, and $\epsilon$-good codes which can also be found in ~\cite{DW}.
Consider the purification, given in Eq. \eqref{eq:input_state}, of the input state $\rho$. Let $Q$ be an $n$-type which is typical, i.e., $||P-Q||_1\leq \delta$ for some $\delta>0$ (see \cite{CoverThomas} for typical sequences). Now consider random variables $u^{(lms)}$ ($l=1,\ldots,L, m=1,\ldots,M, s=1,\ldots,S$), independent identically distributed with respect to the uniform distribution on the type class $\mathcal{T}^n_Q$. Let $\sigma_Q:=\frac{1}{|\mathcal{T}^n_Q|}\sum_{x^n\in \mathcal{T}^n_Q} \rho^E_{x^n}$.
Now we are set to define $\epsilon$-evenness, $\epsilon$-secrecy, and $\epsilon$-good codes ( $\mathcal{C}_l=(u^{(lms)})_{m,s}$) for random variables $u^{(lms)}$. \\
\textbf{$\epsilon$-evenness:} For all $x^n \in \mathcal{T}^n_Q$,
\begin{equation}
    (1-\epsilon)\frac{LMS}{|\mathcal{T}^n_Q|}\leq \sum_{lms}1_{u^{(lms)}(x^n)}\leq(1+\epsilon)\frac{LMS}{|\mathcal{T}^n_Q|},
\end{equation}
where $1_{u^{(lms)}}$ is the indicator function on $\mathcal{T}^n_Q$.\\
\textbf{$\epsilon$-secrecy:} For all $l,m$, the average of Eve's reduced state over the random variable $s$ is close to $\sigma(Q)$, i.e.,
\begin{equation}
    \Bigg|\Bigg| \frac{1}{S}\sum_s\rho^E_{u^{(lms)}}-\sigma(Q)\Bigg|\Bigg|_1\leq \epsilon.
\end{equation}
\textbf{Codes $\mathcal{C}_l$ are $\epsilon$-good:}
The code $\mathcal{C}_l$ is $\epsilon$-good if there exists a POVM $(D^{(l)}_{ms})_{m,s}$ such that
\begin{equation}
    \frac{1}{MS}\sum_{ms}\tr(\rho^B_{u^{(lms)}}D^{(l)}_{ms})\geq 1-\epsilon.
\end{equation}

\subsection{Devetak-Winter $\mathrm{Key}\rightarrow \mathrm{PRand}$ sequential protocol}\label{supp0}

There exists a one-way (from A to B) LOCC  $(n,\epsilon)$-protocol which distills entanglement from $n$ copies of $|\psi_{ABE}\>$ given in Eq. \eqref{eq:input_state}, 
at rate $r_{ent}:=-S(A|B)_\rho$~\cite{DW}. 
In order to show the existence of such a $(n,\epsilon)$-protocol, it is first shown that for each typical type $Q$ (i.e., $||P-Q||_1\leq \delta$ for $\delta>0$) and for large enough $n$, there exists a collection of codewords, $u^{(lms)}$, $l=1,\ldots,L$, $m=1,\ldots,M$, and $s=1,\ldots,S$ from typical sequence, $\mathcal{T}^n_Q$, which satisfy $\epsilon$-evenness, $\epsilon$-secrecy, and a fraction of at least $1-2\epsilon$ of the codes $C_l=(u^{(lms)})_{m,s}$ are $\epsilon$-good (see supplementary material \ref{supp5} for the definitions). The following choices of $L$, $M$, and $S$ are sufficient to show such an existence~\cite{DW}:
\begin{eqnarray}
    L&=&e^{n(H(Q)-I(Q;B)+2\delta)},\nonumber \\
    M&=&e^{n(I(Q;B)-I(Q;E)-3\delta)}, \nonumber \\
    S&=&e^{n(I(Q;E)+2\delta)},
\end{eqnarray}
where $H(Q)$ denotes the  entropy of type $Q$ and $I(Q;B)$ or $I(Q;E)$ are the mutual information of the reduced classical quantum state after measuring type $Q$ on the state $|\psi_{ABE}^{\otimes n}\rangle$.

We can now proceed to describe the protocol which consists of the following steps:
\begin{enumerate}
\item Alice measures  non-destructively the type Q of her subsystem. She aborts the protocol, if the type is atypical. 
\item Performs a quantum operation which outputs classical label $l$ which is the number of a code $C_l$ and decodes $|x\>$ locally into $|ms\>_{M_AS_A}$. In addition, the operation outputs an additional result $\emptyset$ with probability $\leq \epsilon$. The output state when the outcomes $l$ occur is given by $\frac{1}{\sqrt{MS}}\sum_{ms}|ms\>_{M_AS_A}\otimes |\phi^{lms}\>_{BE}$.

\item Alice tells $l$ via public channel to Bob.  Bob applies decoding operation of the system since the code is $\epsilon$-good. Introducing an auxiliary system, $B'$, he applies a unitary on system $BB'$ such that the joint state takes the form $|\vartheta^l\>_{M_AS_AM_BS_BB'E}$, where 
\begin{eqnarray}\label{mil}
  |\vartheta^l&&\>_{M_AS_AM_BS_BB'E}=\frac{1}{\sqrt{MS}}\sum_{ms}|ms\rangle_{M_AS_A}\otimes  \nonumber \\
&&\left(\sqrt{1-e_{ms}}|ms\>_{M_BS_B}|\varphi^{lms}_{OK}\>_{B'E}+\sqrt{e_{ms}}|\varphi^{lms}_{bad}\>_{BB'E}\right) 
\end{eqnarray}
where the $B$ system is split in the subsystems $M_B$ and $S_B$, 
$e_{ms}$ is the probability of misidentifying $ms$, and $|\varphi^{lms}_{bad}\>$ is orthogonal to $|ms\>|\varphi^{lms}_{OK}\>$. 
\item Alice measures register $S_A$ 
in the Fourier basis, $\{|t\>\}$,  where $|t\rangle=\frac{1}{\sqrt{S}}\sum_s \exp(i2\pi ts/S) |s\rangle$ for $t=1,2,\dots S$, and tells the outcome $t$ to Bob.
\label{eq:Fourier}
\item Bob corrects the phase error induced by A by performing a unitary $U:=U_{t}$ on $S_B$ system.
\label{eq:B}
\item B decouples system $M_B$ from $B'E$ by means of unitary $V=\sum_{m}\ket{m}\bra{m}\otimes V^{(m)}_{S_BB'}$.
\label{eq:amplification}
\end{enumerate}
Inspired by results of \cite{YHW19}, we show, how to easily turn the above protocol into a waste-managing one by distilling key plus private randomness instead of distilling entanglement. We call such protocol as Devetak-Winter $\mathrm{Key}\rightarrow \mathrm{PRand}$ sequential protocol, since we extract key and private randomness instead of distilling entanglement as in the \cite{DW}.
\begin{itemize}
\item The first three steps of the above protocol are identical.
\item In the step \ref{eq:Fourier} A does not tell the result $\hat{t}$ of the Fourier measurement. This result is her private randomness. 
\item The step \ref{eq:B} is omitted as B does not know how to correct the errors induced by A's measurement of $S_A$ system in previous step.
\item The step \ref{eq:amplification} is omitted.
\item In the final step A and B also measure their respective systems, $M_A$ and $M_B$, in the computational basis.
\end{itemize}

\begin{lemma}\label{lemma1}
Consider an arbitrary state $\rho_{ABCDE}$. Let us define $\rho^{cccq}:=\sum_{a,b,c}|abc\rangle\langle abc| \langle abc| \tr_{D}\rho_{ABCDE}|abc\rangle$. Consider unitaries, $U_{CD}=\sum_{c}|c\rangle\langle c|\otimes U^c$ and $V_{BD}=\sum_b |b\rangle\langle b|\otimes V^b$. Define state, $\rho^{ps}:=\sum_{a,b,c}|abc\rangle\langle abc| \langle abc| \tr_{D}\left(V_{BD}U_{CD}\rho_{ABCDE}U^\dagger_{CD}V^\dagger_{BD}\right)|abc\rangle$. Then
   $ \rho^{ps}=\rho^{cccq}$.
\end{lemma}
\begin{proof}
\begin{eqnarray}
\rho^{ps}&=&\sum_{a,b,c}\sum_{b',c',b'',c''} |abc\rangle\langle abc|\langle abc|b'c'\rangle\langle b'c'| \nonumber  \\
&&\tr_{D}\left(V^{b'}U^{c'}\rho_{ABCDE}U^{c''\dagger}V^{b''\dagger}\right)|b''c''\rangle\langle b''c''|  abc\rangle \nonumber \\
&=&\sum_{a,b,c}|abc\rangle\langle abc| \langle abc| \tr_{D}\left(V^{b}U^{c}\rho^{ABCDE}U^{c\dagger}V^{b\dagger}\right)|abc\rangle \nonumber \\
&=&\sum_{a,b,c}|abc\rangle\langle abc| \langle abc| \tr_{D}\rho^{ABCDE}|abc\rangle \nonumber \\
&=&\rho^{cccq}. \nonumber
\end{eqnarray}
\end{proof}
Lemma \ref{lemma1} proves that the application of control unitaries of the type $U_{CD}$ and $V_{BD}$ and then performing measurements on $ABC$ in the control basis an arbitrary state, followed by tracing out $D$ will not change the final state $\rho^{cccq}$.
We are ready to provide the proof that the modified Devetak-Winter $\mathrm{Key}\rightarrow\mathrm{PRand}$ sequential protocol provides the honest parties A and B with both private key and private randomness. 

\begin{proposition}
   The final state at the end of our one-way $(n,\epsilon )$  private key and private randomness sequential distillation protocol is such that it is $1-12\sqrt{\epsilon}$ close in fidelity to the ideal state $\frac{1}{MS}\sum_{m,t}|mmt\rangle\langle mmt|_{M_AM_BS_A}\otimes \sigma_E$, where $|t\rangle:=\frac{1}{\sqrt{S}}\sum_s \exp(i2\pi ts/S) |s\rangle$ for $t=1,2,\dots S$.
\end{proposition}
\begin{proof}
The joint state shared by A, B, and E after the third step of our protocol is $|\vartheta^l\rangle_{M_AM_BS_AS_BB'E}$ given in Eq. \eqref{mil}. Here we assume, that system $E$ stores not only initial quantum side information, but also a classical message $|l\>$, i.e. the ket $|l\>$ is a part of a state $|\varphi^{lms}_{OK}\>$. Since the code $C_l$ are $\epsilon$-good, this implies 
\begin{equation}\label{cricket}
    F\left(|\vartheta^l\rangle, \sqrt{\frac{1}{MS}}\sum_{ms}|mm\rangle_{M_AM_B}|ss\rangle_{S_AS_B}|\phi^{lms}_{OK}\rangle_{B'E}\right)\geq 1-3\sqrt{\epsilon}, 
\end{equation}
for each $l$.
For the proof of the above, see supplementary material \ref{app:3}.
where $F(x,y)$ denotes the fidelity between states $x$ and $y$. 
Now lemma \ref{lemma1}, shows that there exists unitary of the form $U_{S_AS_B}=\sum_{t,s}|t\rangle\langle t|\otimes e^{i2\pi ts/S} |s\rangle\langle s|$, where $\{|t\rangle:=\frac{1}{\sqrt{S}}\sum_s \exp(i2\pi ts/S) |s\rangle\}_{t=1}^S$ is a Fourier transform of the basis $\{|s\rangle\}_{s=1}^S$, application of which will not change the final state in our protocol.
Let $|\Xi^l\>:=U_{S_AS_B}\frac{1}{\sqrt{MS}}\sum_{ms}|mm\rangle_{M_AM_B}|ss\rangle_{S_AS_B}|\phi^{lms}_{OK}\rangle_{B'E}=\frac{1}{\sqrt{MS}}\sum_{m,t}|mmt\>_{M_AM_BS_A}|\phi^{lm}\>_{S_BB'E}$ where $|\phi^{lm}\>_{S_BB'E}=|\phi^{lm}\>_{B'E}:=\frac{1}{\sqrt{S}}\sum_s|s\>_{S_B}|\phi^{lms}_{OK}\>_{B'E}$ (absorbing register $S_B$ in $B'$). 
Since unitary operation will keep the fidelity of the states same, therefore for all $l$, we have
\begin{eqnarray}\label{eq:pehla}
F\left(U_{S_AB'}|\vartheta^l\rangle, |\Xi^l\>\right)\geq 1-3\sqrt{\epsilon}.  
\end{eqnarray}
Consider the reduced state of  $\frac{1}{\sqrt{MS}}\sum_{m,t}|mmt\>_{M_AM_BS_A}|\phi^{lm}\>_{B'E}$ on $E$ register be $\sigma^{lm}$. That is, $\sigma^{lm}=\frac{1}{S}\sum_s \tr_{B'}|\phi^{lms}\>\<\phi^{lms}|$. By $\epsilon$-secrecy, the reduced state of $|\Xi^l\>$ in the Eve register must be at a trace distance at most $\epsilon$ from a state $\sigma_E$, independent of indices $l$ and $m$, i.e. $\tr|\sigma_{lm}-\sigma_E|\leq \epsilon$ for each $l$ and $m$. Using the inequality $1-\sqrt{F(\sigma_{lm},\sigma_E)}\leq \frac{1}{2}\tr|\sigma_{lm}-\sigma_E|$, we reach $F(\sigma_{lm},\sigma_E)\geq 1-\epsilon$. 
Notice that in our actual protocol $U_{S_AB'}$ is never applied; however, application of $U_{S_AB'}$ does not change Eve's reduced state and thus the $\epsilon$-secrecy argument can still be applied.
Let the purification of $\sigma_E$ be $|\zeta\>_{B'E}$. Then there must exist unitaries $V^{lm}$ on $B'$ such that 
\begin{eqnarray}
    F(V^{lm}\otimes I_E|\phi^{lm}\>,|\zeta\>)\geq 1-\epsilon.
\end{eqnarray}
Let $V_{M_BB'}:=\sum_{m} |m\>\<m|\otimes V^{lm}$. And $V^{lm}$ can always be chosen such that
\begin{eqnarray}\label{eq:dusra}
    F\left(V_{M_BB'} |\Xi^l\>, \frac{1}{\sqrt{MS}}\sum_{m,t}|mmt\>|\zeta\>\right)\geq 1-\epsilon.
\end{eqnarray}
See supplementary material \ref{supp:chowtha} that combining \eqref{eq:pehla} and \eqref{eq:dusra}, we have
\begin{eqnarray}\label{eq:haire}
    F\left(V_{M_BB'}U_{S_AB'}|\vartheta^l\>,\frac{1}{\sqrt{MS}}\sum_{m,t}|mmt\>|\zeta\>\right)\geq 1-12\sqrt{\epsilon}.
\end{eqnarray}
Let us define a state produced after measuring $M_A$ in the basis $\{|m\>\}$, $M_B$ in the basis $\{|m\rangle\}$ and $S_A$ in the basis $\{|t\rangle\}$ and tracing out $B'$ on the state $V_{M_BB'}U_{S_AB'}|\vartheta^l\>$ as,
\begin{eqnarray}
    &&\rho^{ps}_{\vartheta^l}:=\frac{1}{MS}\sum_{m,t}|mmt\rangle\langle mmt|_{M_AM_BS_A} \nonumber \\ 
    _{M_AM_BS_A}&&\langle mmt| \tr_{B'}\left(V_{M_BB'}U_{S_AB'}\rho_{\vartheta^l} U^\dagger_{S_AB'}V^\dagger_{M_BB'}\right)|mmt\rangle_{M_AM_BS_A} \nonumber,
\end{eqnarray}
where $\rho_{\vartheta^l}=|\vartheta^l\>\<\vartheta^l|$.
Such measurements and tracing out  $B'$ are fidelity non-decreasing operation between two states and therefore from \eqref{eq:haire}, we get
\begin{equation}
    F\left(\rho^{ps}_{\vartheta^l},\frac{1}{MS}\sum_{m,t}|mmt\>\<mmt|_{M_AM_BS_A}\otimes\sigma_E\right)\geq 1-12\sqrt{\epsilon}.
\end{equation}
Now, the final state for given $l$, after DW $\mathrm{Key}\rightarrow \mathrm{PRand}$ sequential protocol is a state, which is obtained by first measuring  $S_A$ in the basis $\{|t\rangle\}$,  
then measuring $M_A$ in the basis $\{|m\}$, $M_B$ in the basis $\{|m\rangle\}$ and tracing out $B'$ on the state $\vartheta^l$ be $\rho^{cccq}_{\vartheta^l}$. 
That is,
\begin{eqnarray}
    \rho^{cccq}_{\vartheta^l}=\frac{1}{MS}\sum_{m,t}|mmt\rangle\langle mmt|_{M_AM_BS_A} \nonumber \\_{M_AM_BS_A}\langle mmt| \tr_{B'}\left(V_{M_BB'}\rho_{\vartheta^l} V_{M_BB'}^\dagger\right)|mmt\rangle_{M_AM_BS_A} \nonumber.
\end{eqnarray}
From Lemma \ref{lemma1}, we know that $\rho^{cccq}_{\vartheta^l}=\rho^{ps}_{\vartheta^l}$. 
Therefore,
\begin{eqnarray}\label{chingari}
F\left(\rho^{cccq}_{\vartheta^l},\frac{1}{MS}\sum_{m,t}|mmt\>\<mmt|_{M_AM_BS_A}\otimes\sigma_E\right)\geq 1-12\sqrt{\epsilon}. \nonumber\\
\end{eqnarray}
The final state after averaging over all $l$ is then given as $\frac{1}{L}\sum_l \rho^{cccq}_{\vartheta^l}$.
Using Eq. \eqref{chingari} and the joint concavity of fidelity, we have
\begin{equation}
    F\left(\frac{1}{L}\sum_l\rho^{cccq}_{\vartheta^l},\frac{1}{MS}\sum_{m,t}|mmt\>\<mmt|_{M_AM_BS_A}\otimes\sigma_E\right)\geq 1-12\sqrt{\epsilon}.
\end{equation}
Therefore the final state after the DW $\mathrm{Key}\rightarrow \mathrm{PRand}$ sequential protocol is close to the ideal state with key in the register $M_A$ and $M_B$ and private randomness in the register $S_A$.  
\end{proof}

\subsection{Proof of inequality \eqref{cricket}}\label{app:3}
The sum of probabilities of extracting $|ms\>_{M_BS_B}$ by Bob from the state $|\vartheta^l\>_{M_AS_AM_BS_BB'E}$  must be higher or equal to $1-\epsilon$ for the code $C_l$ to be $\epsilon$-good which implies
\begin{eqnarray}\label{eq:asaan}
  &&  \sum_{ms}\frac{1}{MS}(1-e_{ms})\geq 1-\epsilon \nonumber \\
\implies &&\sum_{ms}\frac{e_{ms}}{MS}\leq \epsilon \nonumber \\
\implies &&\sum_{ms:e_{ms}>\sqrt{\epsilon}}\frac{e_{ms}}{MS}\leq \epsilon \nonumber \\
    \implies &&\sum_{ms:e_{ms}>\sqrt{\epsilon}}\frac{\sqrt{\epsilon}}{MS}\leq \epsilon \nonumber \\
    \implies &&\sum_{ms:e_{ms}>\sqrt{\epsilon}}\frac{1}{MS}\leq \sqrt{\epsilon}
\end{eqnarray}
That is, at most a fraction of $\sqrt{\epsilon}$ of $e_{ms}$ are greater than the $\sqrt{\epsilon}$. 
Now from \eqref{cricket},
\begin{eqnarray}
    &&F\left(|\vartheta^l\rangle, \sqrt{\frac{1}{MS}}\sum_{ms}|mm\rangle_{M_AM_B}|ss\rangle_{S_AS_B}|\phi^{lms}_{OK}\rangle_{B'E}\right) \nonumber \\&=&\left(\frac{1}{MS}\sum_{ms}\sqrt{1-e_{ms}}\right)^2 \nonumber \\
    &=& \left(\frac{1}{MS}\sum_{ms:e_{ms}>\sqrt{\epsilon}}\sqrt{1-e_{ms}}+\frac{1}{MS}\sum_{ms:e_{ms}\leq\sqrt{\epsilon}}\sqrt{1-e_{ms}}\right)^2 \nonumber\\
    &\geq& \left(\frac{1}{MS}\sum_{ms:e_{ms}\leq\sqrt{\epsilon}}\sqrt{1-e_{ms}}\right)^2\nonumber \\
    &\geq& (1-\sqrt{\epsilon})^3 \geq 1-3\sqrt{\epsilon},
\end{eqnarray}
where \eqref{eq:asaan} is used for the second inequality above.

\subsection{Proof of \eqref{eq:haire}}\label{supp:chowtha}
\begin{lemma}
    Let $\rho, \sigma, \tau$ be quantum states. Then for each $\delta \in (0, 1)$, if $F(\rho, \tau) \geq 1-\delta$ and $F(\sigma, \tau) \geq 1-\delta$, then $F(\rho, \sigma) \geq 1-4\delta$.
\end{lemma}
\begin{proof}
    An \emph{angle} between quantum states $\rho, \sigma$ is defined as $\arccos{F(\rho, \sigma)}$. It is known (see Chapter 9.2 \cite{nielsen2010quantum}) that it satisfies the \emph{triangle inequality}, i.e. for each $\rho, \sigma, \tau$
    \begin{equation}
        \arccos{F(\rho, \sigma)} \leq \arccos{F(\rho, \tau)} + \arccos{F(\sigma, \tau)}.
    \end{equation}
    Since $F(\rho, \tau), F(\sigma, \tau) \geq 1-\delta$ and $\arccos(\cdot)$ is a decreasing function, we can write
    \begin{eqnarray}
        &\arccos{F(\rho, \sigma)} \leq \arccos{F(\rho, \tau)} + \arccos{F(\sigma, \tau)} \leq\\
        &\qquad \qquad \arccos{(1-\delta)} + \arccos{(1-\delta)} = 2\arccos{(1-\delta)}\nonumber,
    \end{eqnarray}
    which implies that
    \begin{eqnarray}
        & F(\rho, \sigma) \geq \cos(2\arccos{(1-\delta)}) = 2[\cos\arccos(1-\delta)]^2 - 1 = \nonumber\\
        & 2 - 4\delta+2\delta^2 - 1 \geq 1-4\delta.
    \end{eqnarray}
\end{proof}
Now, let the states $\rho,\sigma,\tau$ be such that:
\begin{eqnarray}
    F(\rho,\tau)&&\geq 1-3\sqrt{\epsilon}, \nonumber \\
    F(\sigma,\tau)&&\geq 1-\epsilon. \nonumber
\end{eqnarray}
Then, since $1-\epsilon \geq 1-3\sqrt{\epsilon}$, we can take $\delta = 3\sqrt{\varepsilon}$ and use above lemma to show that
\begin{equation}
    F(\rho,\sigma) \geq 1 - 4\cdot 3\sqrt{\epsilon} = 1-12\sqrt{\epsilon}.
\end{equation}

\subsection{Private randomness after private key distillation for isotropic states}\label{sec:private-randomness-isotropic-case}
Here, we estimate the rate of DW
${\mathbf Key}\rightarrow\mathbf{ Prand}$ sequential protocol for 
the two qubit isotropic state.
We parametrize it here as follows:
$\rho_{\mathrm{iso},2}=p\ket{\psi_+}\bra{\psi_+}+(1-p)\frac{\mathbb{I}}{4}$. We denote the rate of private randomness by $\mathrm{Rate}_{\mathrm{Key}_{\mathrm{DW}}\rightarrow \mathrm{PRand}}$.

\begin{lemma}
The rate of private randomness (given positive rate of key) in the case of isotropic state of local dimension $d=2$ reads
    \begin{align}
        &\mathrm{Rate}_{\mathrm{Key}_{\mathrm{DW}}\rightarrow \mathrm{PRand}} = I(X:E)_{\rho_{iso}} =\\
        &-\frac{1+3p}4 \log_2 \left(\frac{1+3p}4\right) -\frac{3(1-p)}4 \log_2 \left(\frac{1-p}4\right) - h\left(\frac{1+p}2\right)\nonumber
    \end{align}
\end{lemma}
\begin{proof}
    We begin with reformulating the isotropic state as folows:
    \begin{align}
        &\rho_{\mathrm{iso},2}=\frac{1+3p}{4}\ket{\psi_+}\bra{\psi_+}+\nonumber\\
        &\frac{(1-p)}{4} \Big(\ket{\psi_-}\bra{\psi_-}+\ket{\phi_+}\bra{\phi_+}+\ket{\phi_-}\bra{\phi_-}\Big),
    \end{align}
    where $\{\ket{\psi_+}, \ket{\psi_-}, \ket{\phi_+}, \ket{\phi_-}\}$ denotes a Bell basis $\{1/\sqrt{2}(\ket{00} + \ket{11}), 1/\sqrt{2}(\ket{00} - \ket{11}), 1/\sqrt{2}(\ket{01} + \ket{10}), 1/\sqrt{2}(\ket{01} - \ket{10})\}$.
    We can then easily find its purification
    \begin{align}
    &|\psi_{ABE}\> = \\
    &\quad M|\psi_{+}\>\otimes |e_{0}\>+N \Big( |\psi_{-}\>\otimes|e_{1}\>+|\phi_{+}\>\otimes|e_{2}\>+|\phi_{-}\>\otimes|e_{3}\>\Big),\nonumber
    \end{align} 
    where $M := \sqrt{\frac{1+3p}{4}}$ and $N:=\sqrt{\frac{(1-p)}{4}}$. We rewrite the above state, to find its Schmidt decomposition in the $A:BE$ cut.
    \begin{align}
        &M|\psi_{+}\>\otimes |e_{0}\>+N \Big( |\psi_{-}\>\otimes|e_{1}\>+|\phi_{+}\>\otimes|e_{2}\>+|\phi_{-}\>\otimes|e_{3}\>\Big) = \nonumber\\
        & \frac{1}{\sqrt{2}}|0\>\Bigg[\sqrt{M^2 + N^2}|0\>\Big(\frac{M}{\sqrt{M^2 + N^2}}|e_0\> + \frac{N}{\sqrt{M^2 + N^2}}|e_1\>\Big) +\nonumber\\
        &\qquad \sqrt{2N^2}|1\>\Big(\frac{N}{\sqrt{2N^2}}|e_2\> + \frac{N}{\sqrt{2N^2}}|e_3\>\Big)\Bigg] + \\
        & \quad \frac{1}{\sqrt{2}}|1\>\Bigg[\sqrt{2N^2}|0\>\Big(\frac{N}{\sqrt{2N^2}}|e_2\> - \frac{N}{\sqrt{2N^2}}|e_3\>\Big) + \nonumber\\
        &\qquad \sqrt{M^2 + N^2}|1\>\Big(\frac{M}{\sqrt{M^2 + N^2}}|e_0\> - \frac{N}{\sqrt{M^2 + N^2}}|e_1\>\Big)\Bigg].\nonumber
    \end{align}
    The above equality can easily be checked by direct calculations. After going back to $\sqrt{\frac{1+3p}{4}}$ and $\sqrt{\frac{(1-p)}{4}}$ instead of $M$ and $N$ we obtain
    \begin{align}
        & \frac{1}{\sqrt{2}}|0\>\Big[\sqrt{\frac{1+p}{2}}|0\>\Big(\alpha|e_0\> + \beta|e_1\>\Big) +\nonumber\\
        &\qquad \sqrt{\frac{1-p}{2}}|1\>\Big(\frac{1}{\sqrt{2}}|e_2\> + \frac{1}{\sqrt{2}}|e_3\>\Big)\Big] + \\
        & \quad \frac{1}{\sqrt{2}}|1\>\Big[\sqrt{\frac{1-p}{2}}|0\>\Big(\frac{1}{\sqrt{2}}|e_2\> - \frac{1}{\sqrt{2}}|e_3\>\Big) + \\
        &\quad \qquad \sqrt{\frac{1+p}{2}}|1\>\Big(\alpha|e_0\> - \beta|e_1\>\Big)\Big],
    \end{align}
    where $\alpha = \sqrt{\frac{1+3p}{2(1+p)}}, \beta = \sqrt{\frac{1-p}{2(1+p)}}$. After the measurement of $A$ system in a computational basis, we obtain the following cqq-state
    \begin{equation}
        \rho_{XBE} = \frac{1}{2}\left(|0\>\<0|_A\otimes \phi_{BE}^{(0)} + |1\>\<1|_A\otimes \phi_{BE}^{(1)}\right),
    \end{equation}
    where $\phi_B^{(0)} = (1+p)/2|0\>\<0| + (1-p)/2|1\>\<1|$ and $\phi_B^{(1)} = (1-p)/2|0\>\<0| + (1+p)/2|1\>\<1|$. Now, it is straightforward to compute the desired quantity $I(X:E)_\rho = H(E) - \frac{1}{2}(H(\phi_{E}^{(0)}) + H(\phi_{E}^{(1)}))$. Indeed, since $H(E) = H(AB)_{\rho_{\mathrm{iso,2}}}$ and $H(\phi_{E}^{(i)}) = H(\phi_{B}^{(i)})$, we obtain
    \begin{align}
        &\mathrm{Rate}_{\mathrm{Key}_{\mathrm{DW}}\rightarrow \mathrm{PRand}} = I(X:E)_\rho = \\
        & H\left(\left\{\frac{1+3p}{4},  \frac{(1-p)}{4}, \frac{(1-p)}{4}, \frac{(1-p)}{4}\right\}\right) - h\left(\frac{1 + p}{2}\right) = \nonumber\\
        &-\frac{1+3p}4 \log_2 \left(\frac{1+3p}4\right) -\frac{3(1-p)}4 \log_2 \left(\frac{1-p}4\right) - h\left(\frac{1+p}2\right),
    \end{align}
    where $h(\cdot)$ denotes binary Shannon entropy.
\end{proof}
In the above lemma, we consider the case, in which the rate of key distillation is positive. If this is not the case, we can use result from the main theorem of \cite{YHW19}, which states that in our case, randomness on $A$ site can distilled at rate $\log|A| - H(A|B)_+$. For an isotropic state considered here, when $H(A|B)>0$, this quantity equals $2 - H(AB)$. On Fig. \ref{fig:iso_gain} we see the randomness distillation rate in both cases. Note that Fig. \ref{fig:iso_gain} uses parameter $f$, which is consistent with the form of isotropic state given in Eq. \eqref{eq:iso_state}. The relation between parameters $f$ and $p$ (which was used in this section) is given by $f(p) = (1+3p)/4$.

\bibliography{output}

\end{document}